\documentclass[runningheads]{llncs}
 
\usepackage{amsthm}
\usepackage{amsfonts}
\usepackage{amsmath}
\usepackage{hyperref}

\usepackage{amssymb}
\usepackage{macros}
\usepackage{xspace}
\usepackage{multirow}
\usepackage{subcaption}
\usepackage{diagbox}
\usepackage{bm}
\usepackage{todonotes}
\usepackage{fullpage}
\usepackage{makecell}
\usepackage{orcidlink}

\usepackage[notion,quotation,makeidx, electronic]{knowledge}

\usepackage{algpseudocode,algorithm2e}

\usepackage{tikz}
\usetikzlibrary{arrows,automata,positioning}

\title{Learning Real-Time One-Counter Automata Using Polynomially Many Queries}

\titlerunning{Learning DROCAs Using Polynomially Many Queries}

\author{Prince Mathew\inst{1}\orcidlink{0000-0001-6410-1474} \and
Vincent Penelle\inst{2} \and
A.V. Sreejith\inst{3}}

\authorrunning{P. Mathew, V. Penelle, and A.V. Sreejith}

\institute{Indian Institute of Technology Goa, India.  \email{prince@iitgoa.ac.in} \and
Univ. Bordeaux, CNRS,  Bordeaux INP, LaBRI, UMR 5800, F-33400, Talence, France. \email{vincent.penelle@u-bordeaux.fr}
\and
Indian Institute of Technology Goa, India.  \email{sreejithav@iitgoa.ac.in} }

\usepackage[noabbrev]{cleveref}
\Crefname{clam}{Claim}{Claims}

\begin{document}
\knowledge{notion, index=counter-synchronous}
| counter-synchronised
| counter-synchronous
| counter-synchronicity
\knowledge{notion, index=$\sgn$}
| \sgn
\knowledge{notion, index=$\P$}
| \P
\knowledge{notion, index=$\S$}
| \S
\knowledge{notion, index=$\Memb$}
| \Memb
\knowledge{notion, index=$\Actions$}
| \Actions
\knowledge{notion, index=$d$-closed}
| $d$-closed
\knowledge{notion, index=$d$-consistent}
| $d$-consistent
\knowledge{notion, index=$\CV\upharpoonright$}
| \CV\upharpoonright
\knowledge{notion, index=row}
| row
\knowledge{notion, index=Operations}
| Operations
\knowledge{notion, index=$\widetilde\Sigma$}
| \widetilde\Sigma
\knowledge{notion, index parent key=one-counter automata, index=$\droca$}
| \droca
| \drocas
\knowledge{notion, index=\minOCA}
| \minOCA
\knowledge{notion, index=\bps}
| \bps
\knowledge{notion, index=\dsOne}
| \dsOne
\knowledge{notion, index=\dsTwo}
| \dsTwo
\knowledge{notion, index parent key=one-counter automata, index=$\voca$}
| \voca
| \vocas
| visibly one-counter automata
| visibly one-counter automaton
\knowledge{notion, index=\encodedDFA}
| \encodedDFA
\knowledge{notion, index=observation table}
| observation table
| observation tables
\knowledge{notion, index=$\simeq$}
| \simeq
| \not\simeq
\knowledge{notion, index=$\not\sim$}
| \not\sim
\knowledge{notion, index= $height$}
| height_{\Autom}
| height_{\Butom}
| height
\knowledge{notion}
| \Lang
| \Autom(w)
| \Autom(u)
| \Autom(uz)
| \Autom(vz)
| \Butom(w)
| \Autom(z)
| \Cutom(z)
\knowledge{notion, index= equivalence}
| equivalence
| equivalent
\knowledge{notion, index=minimal synchronous-equivalence query}
| minimal synchronous-equivalence queries
| minimal synchronous-equivalence query
| MSQ
\knowledge{notion, index=membership query}
| membership queries
| membership query
| membership
| MQ
\knowledge{notion, index= $\ce$}
| \ce
| counter-effect
\knowledge{notion, index=$\Enc_{\Autom}$}
| \Enc
\knowledge{notion, index=partial equivalence query}
| partial equivalence query
| partial equivalence queries
\knowledge{notion, index=counter value query}
| counter value queries
| counter value query
| counter value 
| CV

\maketitle
\begin{abstract}
In this paper, we introduce a novel method for active learning of deterministic real-time one-counter automaton ("\droca"). The existing techniques for learning a "\droca" rely on observing the behaviour of the "\droca" up to exponentially large counter values. Our algorithm eliminates this need and requires only a polynomial number of queries. {Additionally, our method differs from existing techniques as we learn a minimal "counter-synchronous" "\droca", resulting in much smaller counter-examples on equivalence queries. Learning a minimal "counter-synchronous" "\droca" cannot be done in polynomial time unless $\CF{P = NP}$, even in the case of "visibly one-counter automata".} {We use a SAT solver to overcome this difficulty.} 
The solver is used to compute a minimal separating \dfa from a given set of positive and negative samples.

We prove that the "equivalence" of two "counter-synchronous" "\drocas" can be checked significantly faster than that of general "\drocas". 
For "visibly one-counter automata", we have discovered an even faster 
algorithm for "equivalence" checking. 
We implemented the proposed learning algorithm and tested it on randomly generated "\drocas". Our evaluations show that the proposed method outperforms the existing techniques on the test set. 
\keywords{One-counter automata \and Active learning \and SAT solver.} \end{abstract}

\section{Introduction}
\label{sec:intro}
{\renewcommand{\thefootnote}{}
\footnotetext{This document contains numerous links to enhance its usability. Terms and concepts defined within the document are directly linked to their definitions as hyperlinks.}}
The problem of identifying a model from a given dataset is an area of interest in various fields of computer science, like formal verification and machine learning. 
However, inferring the right model from labelled samples is challenging. For instance, finding a minimal separating \dfa ~-- a \dfa that accepts a given set of positive samples and rejects a given set of negative samples -- is known to be $\CF{NP}$-complete~\cite{gold}. 
Angluin \cite{angluin} introduced an active learning framework called $\intro[\Lstar]{}\Lstar$, involving a learner and a teacher. The learner can ask membership and equivalence queries to the teacher. Angluin showed that \dfa can be learnt in polynomial time {with respect to the size of the minimal \dfa}.

In this paper, we are interested in active learning of a deterministic real-time one-counter automaton ("\droca"). These are finite-state machines equipped with a non-negative integer counter that can be incremented or decremented on reading an input symbol. 
The counter adds expressive power, enabling a "\droca" to recognise certain non-regular context-free languages (e.g., $\{a^nb^n\mid n>0\}$). However, this added power also introduces significant challenges in learning. 

\paragraph{Our contribution.}
We introduce the notion of "counter-synchronous" "\drocas"~-- two "\drocas" are "counter-synchronous" if, for any word, the counter value reached on reading that word is the same on both machines. 
{Given two "\drocas" with $\K$ states, }
we give an $\mathcal{O}(\alpha(\K^5)\K^5)$ time\footnote{ For all practical applications, one can consider $\alpha$ as a constant~(see \Cref{sec:equivalence}).} algorithm that solves the following two problems (see \Cref{countersync}): (1) check if {they} are "counter-synchronous", and (2) check {whether they }are "equivalent", {if they are "counter-synchronous"}.
For "visibly one-counter automata" ("\droca" where the input alphabet determines the counter-actions), we have devised an even faster $\mathcal{O}(\alpha(\K^3)\K^3)$ algorithm for checking "equivalence" (see \Cref{vocaeq}).  

The main result of this paper is a novel approach for active learning of "\drocas". However, the active learning framework differs from that introduced by Angluin in a few crucial aspects (see \Cref{LearnDroca}).
Similar to the work by Bruyère et al.~\cite{gaetan}, we use an additional query called "counter value query". This allows the learner to ask for the counter value reached on reading a word in the "\droca". Furthermore, the learner has access to a "minimal synchronous-equivalence query" on the "\droca". The teacher returns true for an equivalence query if the learnt "\droca" is "counter-synchronous" and "equivalent" to the teacher's "\droca". Otherwise, it returns a minimal word that violates this property. 

In this framework, we give an algorithm that learns a minimal "counter-synchronous" "\droca". A key innovation in our approach is the use of a SAT solver for solving the $\CF{NP}$-hard problem of finding a minimal separating \dfa from a set of positive and negative samples. The solver, in conjunction with a modified version of $\Lstar$, learns a "\encodedDFA" (see \Cref{chardfa}). Subsequently, we use this "\encodedDFA" to construct a minimal "counter-synchronous" "\droca". Our algorithm requires only a polynomial number of queries to the SAT solver and the teacher. Consequently, our algorithm is in $\CF{P^{NP}}$.

\paragraph{Justification for using a SAT solver.}
We argue that unless $\CF{P=NP}$, learning a minimal "visibly one-counter automaton" ("\voca") cannot be done in polynomial time. 
 This follows from the fact that minimisation of "\voca" can be reduced to learning a minimal "\voca". Furthermore, It was pointed out by Michaliszyn and Otop~\cite{learningVPDA} that minimising "\voca" is $\CF{NP}$-complete. This follows from the result by Gauwin et al.\cite{minVPDA} that minimising \textsc{vpda} is $\CF{NP}$-complete.

\paragraph{Comparison with existing methods.}
From a complexity theoretical perspective, "\drocas" can be learned with polynomial space and exponential time with a straightforward brute-force approach. This method entails enumerating all conceivable "\drocas", starting with a one-state "\droca", and submitting equivalence queries for each. This approach, without a doubt, entails an exponential number of equivalence queries.
All existing algorithms for learning "\drocas", including the algorithm by Fahmy and Roos~\cite{FahRoo}, the algorithm by Bruyère et al.~\cite{gaetan}, {and the algorithm for learning "\vocas" by  Neider and L{\"{o}}ding~\cite{christof}} require exponential time and an exponential number of queries with respect to the number of states of a minimal "\droca" recognising the language.  
All these algorithms share the idea of learning the initial portion of an infinite behavioural graph and then seek to identify a repetitive structure in it. However, in the worst-case scenario, this repetitive structure becomes apparent only after learning an exponentially large portion of the graph. {In this case, the learnt "\droca" will be exponentially large.} 
Consequently, learning this exponential-sized behaviour necessitates exponentially many queries. 
{Moreover, the equivalence queries also run on these exponentially large "\drocas", making it even more infeasible.} 
Bruyère et al.~\cite{gaetan} were the first to pursue a practical learning application of learning "\drocas". However,  due to the difficulty in checking the "equivalence" of "\drocas", they had to use a weaker form of equivalence query that checks for counter-examples up to some random counter value. Their "equivalence" check might say two non-equivalent "\drocas" are "equivalent" if the minimal counter-example is large. 

Our approach ("\minOCA") differs fundamentally from these existing methods by eliminating the need to observe the automaton's behaviour up to exponentially large counter values. 
{We propose an algorithm for learning "\drocas" using only a polynomial number of queries. This sets it apart from existing techniques that require exponentially many queries for learning.} 
 One significant bottleneck in learning "\drocas" is the "equivalence" test by the teacher. Given two "\drocas" with number of states less than some $\K\in\N$, the "equivalence" check takes $\mathcal{O}(\K^{26})$ time\footnote{This polynomial is derived from the algorithm given in \cite{droca}.}. 
 This is impractical for real-world applications. To mitigate this, we use the synchronous-equivalence check that runs in $\mathcal{O}(\K^5)$ time. We obtain significantly smaller counter-examples while using this "equivalence" check. 
Our equivalence queries are also on models whose size is less than or equal to a minimal "counter-synchronised" "\droca". 
 Furthermore, unlike existing techniques that learn exponentially large "\drocas", our algorithm always learns an equivalent "counter-synchronous" "\droca" with the minimal number of states. 
 
\paragraph{Experiments.}
We evaluate an implementation of our algorithm 
and compare the results obtained with the existing technique by Bruyère et al.~\cite{gaetan}. Experiments were conducted on randomly generated "\drocas" with number of states ranging from 2 to 15 and the input alphabet size varying from 2 to 5. The results indicate that the proposed method outperforms the existing one~\cite{gaetan}. 

{The remainder of this paper is organised as follows: \Cref{sec:definitions} gives the definitions and preliminaries. \Cref{sec:equivalence} gives "equivalence" results of "counter-synchronised" "\drocas". \Cref{LearnDroca} details our learning algorithm for "\drocas", \Cref{sec:implementation} covers the implementation details and presents our experimental results. Finally, \Cref{sec:conclusion} summarises our work and suggests future research directions.}

\section{Definitions and Preliminaries}
\label{sec:definitions}
For any finite set $S$, $|S|$ denotes its cardinality. Non-negative numbers are denoted by $\N$, and 
$[i,j]$ denote the interval $\{i, i+1, \ldots , j\} \subseteq \N$. 
For any $d\in\N$, the sign of $d$ (denoted by $""\sgn""(d)$) is defined as $"\sgn"(d)=0$ if $d=0$ and is $1$ otherwise.
Let $w=a_0a_1a_2\ldots a_n \in \Sigma^*$. For $j,k\in[0,n]$, with $j<k$, we use $w[j]$ to denote the letter $a_j$ and $w_{[j \cdots k]}$ to denote the factor $a_j a_{j+1}\cdots a_{k}$. 

\begin{definition}
\label{defdroca}
A deterministic real-time one-counter automaton \text{\upshape(""\droca"")} $\Autom= (Q,\Sigma, q_0,\delta_0,\delta_1,F)$, where $Q$ is a finite nonempty set of states, $\Sigma$ is the input alphabet, $q_0\in Q$ is the initial state, $\delta_0: Q \times \Sigma \to Q \times \{0,+1\}$ and $\delta_1: Q \times \Sigma \to Q \times \{0,+1,-1\}$ are the transition functions, and $F\subseteq Q$ is the set of final states.
\end{definition}

We use $|\Autom|$ to denote the size of $\Autom$, which we consider to be $|Q|$. 
A configuration $\con$ of a "\droca" is a pair $(q, n)\in Q \times \N$, where $q\in Q$ denotes the current state and $n\in\N$ is the counter value. The configuration $c_0=(q_0,0)$ is called the \emph{initial configuration} of $\Autom$. 
Let $p,q \in Q, n\in\N, e \in \{-1,0,+1\}$ and $a \in\Sigma$.
A \emph{transition} between two configurations $(p,n)$ and $(q, n+e)$ on the symbol $a$ is defined, 
if $\delta_{"\sgn"(n)}({p},a) = ({q},e)$. We use $(p,n)\xrightarrow{a}(q,n+e)$ to denote this. 
A run on a word $w=a_1\dots a_n$ from a configuration $(p_1,m_1)$ is the sequence of transitions $(p_1,m_1) \xrightarrow{a_1} (p_2,m_2) \xrightarrow{a_2} (p_3,m_3) \xrightarrow{a_3} \dots \xrightarrow{a_{n-1}} (p_n,m_n)$ where $p_i \in Q$ and $m_i$s are counter values. In this case, we denote the run by $(p_1,m_1) \xrightarrow{w} (p_n,m_n)$. 
Note that the counter values always stay non-negative, implying a decrement is not permitted from a configuration with zero counter value.

\AP
Let $q \in Q,\ n\in\N$, and $w \in\Sigma^*$ with $(q_0,0)\xrightarrow{w}(q,n)$. We use $\intro[\ce]{}\ce_\Autom(w)= n$ to denote {the counter value reached on reading $w$ from the initial configuration. We call $\ce_\Autom(w)$ the \emph{"counter-effect"} of $w$. We denote by $""height_{\Autom}""(w)$ the maximal} "counter-effect" of the prefixes of $w$ in \Autom. We drop the subscript $\Autom$ when the $\droca$ under consideration is evident. 
$\Autom$ is deterministic (resp. complete) if, for any word $w$, there is at most (resp. at least) one run on $w$ starting from any configuration. 
We say that a word $w$ is accepted by $\Autom$ if and only if $(q_0,0) \xrightarrow{w} (q_f,m)$ for some $q_f \in F$ and $m \in \N$. The language of $\Autom$, denoted by $""\Lang""(\Autom)$, is the set of all words accepted by $\Autom$.
For convenience, we use $"\Autom(w)"=1$ if $w\in\Lang(\Autom)$ and $"\Autom(w)"=0$ otherwise. Two "\drocas" $\Autom$ and $\Butom$ are ""equivalent"" if $"\Lang"(\Autom)="\Lang"(\Butom)$.

There are no $\epsilon$-transitions in a "\droca". {We consider only complete "\drocas" in this paper.  Every incomplete "\droca" can be made complete without changing its language by directing all the undefined transitions to a new non-final sink state.}

\AP
\begin{definition}
We say that two "\drocas" \Autom and \Butom are ""counter-synchronised"" if for all $w\in\Sigma^*$, \upshape$\ce_{\Autom}(w)=$\upshape$\ce_{\Butom}(w)$.
\end{definition}
  
\section{Equivalence of Counter-Synchronised DROCAs}
\label{sec:equivalence}
The "equivalence" of "\drocas" was shown to be in $\CF{NL}$ by B{\"{o}}hm and G{\"{o}}ller~\cite{droca}. They show that there exists a $\mathcal{O}(\K^{26})$ length word that distinguishes two non-equivalent "\drocas" of size $\K$. 
However, this is not suitable for practical applications. 

We give an efficient $\mathcal{O}(\K^5)$ algorithm to check the "equivalence" of "counter-synchronised" "\drocas". {If two "\drocas" are not "counter-synchronised", then our algorithm outputs the smallest word for which they reach different counter values.
If they are "counter-synchronised" and not "equivalent", then our algorithm outputs the smallest word that is accepted by one and rejected by the other. In both these cases, the length of the counter-example is $\mathcal{O}(\K^5)$. 

{In \Cref{countersync,vocaeq}, we can treat $\alpha$ as a constant since $\alpha(j) \leq 4$ for all $j<2^{1000}$ (see \cite{almeida}). Hopcroft and Karp~\cite{hopcroft} showed that "equivalence" checking of two \dfas of size $n$ can be done in $\mathcal{O}(\alpha(n)n)$ time.
\begin{theorem}
\label{countersync}
Let \Autom and \Butom be "\drocas" where $|\Autom|,|\Butom|\leq\K$. {If \Autom is not "counter-synchronised" and "equivalent" to \Butom, then there is a minimal word $w$ where $|w|\leq 2\K^5$, $"height_{\Autom}"(w),\ "height_{\Butom}"(w)\leq \K^4$, and either \upshape$\ce_{\Autom}(w)\neq$\upshape$\ce_{\Butom}(w)$ $or$ $"\Autom(w)"\neq "\Butom(w)"$.
\textit{There is an $\mathcal{O}(\alpha(\K^5)\K^5)$ time algorithm to output this word if it exists.}} 
\end{theorem}
\begin{proofsketch}
Let \Autom and \Butom be "counter-synchronous" but not "equivalent". The "equivalence" of \Autom and \Butom reduces to reachability in their product "\droca". From \cite{shortestpaths}, the "height" of a minimal word that reaches a configuration is $\K^4$ for a "\droca" of size $\K^2$.

Let \Autom and \Butom be not "counter-synchronous" and $w$ be a minimal word where $\ce_{\Autom}(w)\neq$\upshape$\ce_{\Butom}(w)$.
Let $w=w_1a$ for some $w_1$ and $a\in\Sigma$. Since $w$ is minimal,  \upshape$\ce_{\Autom}(w^\prime)=$\upshape$\ce_{\Butom}(w^\prime)$ for all prefixes of $w'$ of $w_1$. The synchronous run of $w_1$ can be seen as the run of a "\droca" with $\K^2$ states.
From \cite{shortestpaths}, $"height_{\Autom}"(w_1) \leq \K^4$.

In both cases, it suffices to search for words of "height" less than $\K^4$. 
We construct \dfas of size $\K^5$ corresponding to the configuration graph of \Autom and \Butom up to counter value $\K^4$. We conclude by doing a \dfa equivalence check. 
\end{proofsketch}

 A "visibly one-counter automaton" (""\voca"") is a ("\droca") where the input alphabet is $\Sigma = (\Sigma_{+}, \Sigma_{-},\Sigma_{0})$. 
 The counter value is incremented (resp.~decremented, not changed) on reading a symbol from $\Sigma_{+}$ (resp.~$\Sigma_{-}$, $\Sigma_0$). Note that "\vocas" over the same alphabet are "counter-synchronous" by definition as the word itself defines the counter value reached. 
For "\vocas", we give a faster $\mathcal{O}(\alpha(\K^3)\K^3)$ time algorithm for "equivalence" checking using ideas from \cite{wodca}.

\begin{figure}
\centering
\resizebox{.4\columnwidth}{!}{%
\begin{tikzpicture}
\tikzset{every path/.style={line width=.3mm}}
\draw plot [smooth, tension=.6] coordinates {(-2,-1)(-1.1,1)(-.3,-1) (0,0) (.3,1) (.7,2) (1.35,1) (1.64,1.5) };
\draw [dashed] plot [smooth, tension=.6] coordinates{(1.64,1.5) (2,2.5)(2.25,3.3)};
\draw plot [smooth, tension=.6] coordinates{(2.25,3.3) (3,5) (3.8,3.3)};
\draw [dashed] plot [smooth, tension=.6] coordinates{(3.8,3.3) (4.2,2.8) (4.7,3.5) (4.95,3.3)(5.2,2.5) (5.53,1.5)};
\draw (3,4.3) node {$\mata_{j}$};
\draw plot [smooth, tension=.6] coordinates{ (5.53,1.5)(5.7,1) (6,0)  (6.6,.8) }; 
\filldraw (3,5) circle[radius=1.5pt] node[anchor=north,yshift=0.6cm] {$d$};
\filldraw (-2,-1) circle[radius=1.5pt] node[anchor=south,yshift=-0.6cm] {$c_{\iota}$};
\filldraw (6.6,.8) circle[radius=1.5pt] node[anchor=south,yshift=-0.6cm] {$c_\ell$};
\filldraw (0.25,.8) circle[radius=1.5pt] node[anchor=south,yshift=-0.6cm] [xshift=-.3cm, yshift=.2cm]{$$};
\filldraw (5.75,.8) circle[radius=1.5pt] node[anchor=south,yshift=-0.6cm] [xshift=-.3cm, yshift=.2cm]{$$};
\filldraw (1.64,1.5) circle[radius=1.5pt] node[anchor=west, yshift=-.2cm] {$(\vc y_{t}, t)$}; 
\filldraw (5.53,1.5) circle[radius=1.5pt] node[anchor=west, yshift=-.2cm] {$(\vc y'_{t},t)$} node[xshift=.54cm, yshift=-1.9cm]{$\matb_t$};
\filldraw (2.25,3.3) circle[radius=1.5pt] node[anchor=south, xshift=-.4cm] {$(\vc y_{j},j)$};
\filldraw (3.8,3.3) circle[radius=1.5pt] node[anchor=south,xshift=.4cm] {$(\vc y'_{j},j)$};
\draw [line width=0.38mm] [-stealth](-2,-1) node[anchor=south, xshift=4cm, yshift=-.7cm]{word length}-- (7,-1);
\draw [gray!70] (-2,5)node[anchor=east]{$m$} -- (7,5) ;
\draw[gray!70](-2,1.5)node[anchor=east]{$t$} -- (7,1.5);
\draw[gray!70](-2,3.3)node[anchor=east]{$j$} -- (7,3.3);
\draw[gray!70](-2,.8)node[anchor=east]{$n_\ell$} -- (7,.8);
\draw [line width=0.38mm] [-stealth](-2,-1)node[anchor=south,rotate=90, xshift=3cm, yshift=.5cm] {counter value} -- (-2,5.8);\end{tikzpicture}
}
\caption{The figure shows the synchronous run of a word on two "\vocas" such that it reaches a final state in one "\voca" and a non-final state in the other. 
The dashed line denotes the part of the synchronous run that can be removed. }
\label{fig:ucut}
\end{figure}

\begin{theorem}
\label{vocaeq}
{Let \Autom and \Butom be "\vocas" where $|\Autom|,|\Butom|\leq\K$. If $\Autom$ and $\Butom$ are not "equivalent", then there is a word $w$ where $"\Autom(w)"\neq "\Butom(w)"$, $|w|\leq 4\K(\K+\K^2)$, and $"height_{\Autom}"(w)\leq 2(\K+\K^2)$.
An $\mathcal{O}(\alpha(\K^3)\K^3)$ algorithm can find this $w$ if it exists.}
\end{theorem}
\begin{proof}
Let $\Autom=(P,\Sigma, p_{\iota},\delta^1_0,\delta^1_1,F_1)$ and $\Butom=(Q,\Sigma, q_{\iota},\delta^2_0,\delta^2_1,F_2)$ be two non-equivalent "\vocas".  {Let $P=\{p_1,\ldots, p_{|\Autom|}\}$ and $Q=\{q_1,\ldots, q_{|\Butom|}\}$.} 
For any word $w$, $\ce_{\Autom}(w)=\ce_{\Butom}(w)$. Let $\ce(w)$ denote this value. Similary, let $"height"(w)$ denote $"height_{\Autom}"(w) = "height_{\Butom}"(w)$.
 For states $p_i \in P$ and $q_j \in Q$, we define the row vector $\vc x_{(p_i,q_j)} \in \{0,1\}^{|\Autom|+|\Butom|}$ as: $\vc x_{(p_i,q_j)}[k] = 1$ if and only if $k=i$ or $k= |\Autom|+j$. We also define the row vector $\vc\eta \in \{-1,0,1\}^{|\Autom|+|\Butom|}$ where $\vc \eta[k] = 1$ if $k \leq |\Autom|$ and $p_k \in F_1$, $\vc \eta[k] = -1$, if $k > |\Autom|$ and $q_{k-|\Autom|} \in F_2$, and $\vc \eta[k]=0$ otherwise.
Therefore, $\vc x_{(p, q)}\trans {\vc \eta} \neq 0$, if exactly one among $p$ and $q$ is a final state. 

We consider the synchronous run of $\Autom$ and $\Butom$. A configuration pair of this synchronous run is $(\vc x_{(p,q)}, n)$ where $p \in P, q \in Q$, and $n \in \N$.
The initial configuration pair is $c_{\iota} = (\vc x_{(p_{\iota},q_{\iota})}, 0)$. Given two configuration pairs $c_1=(\vc x_{(p,q)},n)$ and $c_2=(\vc x_{(p',q')},m)$, $c_1 \xrightarrow{u} c_2$ denotes
$(p,n) \xrightarrow{u} (p',m)$ and $(q,n) \xrightarrow{u} (q',m)$. \final{The transition matrix on $u$ from $c_1$ is $\matm\in\{0,1\}^{(|\Autom|+|\Butom|)^2}$ where $\matm[i,j]=1$ if and only if $(p_i,n) \xrightarrow{u} (p_j,m)$ or $(q_{i-|\Autom|},n)\xrightarrow{u}(q_{j-|\Autom|},m)$. 
Hence, $\vc x_{(p,q)} \matm = \vc x_{(p',q')}$. \Cref{matrix1,matrix2,matrixDep} follow.}

\begin{clam}
\label{matrix1}
For any $p,p' \in P$, $q,q' \in Q$ and word $w$, the transition matrix of $w$ from $(\vc x_{(p,q)},n)$ is the same as the transition matrix of $w$ from $(\vc x_{(p',q')},n)$.
\end{clam}
\begin{clam}
\label{matrix2}
Let $w$ be such that $\ce(w') > 0$ for all prefixes $w'$ of $w$. Then, the transition matrix of $w$ from $(\vc x_{(p,q)},n)$ is the same as the transition matrix of $w$ from $(\vc x_{(p',q')},m)$ for any $p,p' \in P$ and $q,q' \in Q$ and $m,n >0$.
\end{clam}

\begin{clam}
Any set of $|\Autom|^2+|\Butom|^2+1$ transition matrices is linearly dependent.
\label{matrixDep}
\end{clam}
Let $w$ be a minimal word such that $"\Autom(w)"\neq "\Butom(w)"$. Let $c_{\ell}=(\vc x_{\ell},n_{\ell})$ be such that $c_{\iota} \xrightarrow{w} c_{\ell}$. Hence, $\vc x_{\ell} \trans{\vc \eta} \neq 0$. {Note that there is no configuration pair that repeats during the run of $w$. Otherwise, we can remove the run between the repetitions to get a shorter word. This will contradict the minimality of $w$.}

\begin{clam}
\label{heightboundVoca}
$"height"(w)\leq n_\ell+|\Autom|^2+|\Butom|^2$.
\end{clam}
\begin{clamproof}
Assume for contradiction that $"height"(w)=m$ and $m> n_\ell + |\Autom|^2+|\Butom|^2$. There is a factorisation of $w = w_1w_2$ such that $c_{\iota}\xrightarrow{w_1} (\vec z,m) \xrightarrow{w_2}c_{\ell}$. 
 For an $i \in [n_{\ell},m]$, let $(\vc y_i,i)$ (resp. $(\vc y_i',i)$) be the configuration pair where the counter value $i$ is encountered for the last (resp. first) time during the run of $w_1$ (resp. $w_2$). Let $\mata_i$ and $\matb_i$ denote the transition matrices where $\vc y_i \mata_i= \vc y_i'$ and $\vc y_i' \matb_i= \vc x_{\ell}$. Let $w = x_iy_iz_i$ where  $c_{\iota}\xrightarrow{x_i} (\vc y_i,i) \xrightarrow{y_i} (\vc y_i',i) \xrightarrow{z_i}c_{\ell}$. See \Cref{fig:ucut}.
Consider the matrices $\mata_{m-1},\mata_{m-2},\ldots, \mata_{n_\ell}$ in order. From \Cref{matrixDep}, there exists $t \in [n_\ell, m-1]$ where $\mata_t$ is linearly dependent on $\mata_{m-1},\cdots,\mata_{t+1}$. Since $\vc y_t \mata_t \matb_t \trans{\vc \eta} \neq 0$, $\vc y_t (r_{t+1}\mata_{t+1}+\cdots+r_{m-1}\mata_{m-1} )\matb_t \trans{\vc \eta} \neq 0$ for integers $r_{t+1},\ldots,r_{m-1}$. Hence there is a $j > t$ where $\vc y_t \mata_j \matb_t \trans{\vc \eta} \neq 0$. 
We show $x_ty_jz_t$ is accepted exactly by one of $\Autom$ and $\Butom$ contradicting minimality of $w$. It suffices to show $c_{\iota} \xrightarrow{x_t} (\vc y_t,t) \xrightarrow{y_j} (\vc y_t \mata_j,t) \xrightarrow{z_t}  (\vc y_t \mata_j \matb_t,n_{\ell})$. From \Cref{matrix2}, the transition matrix of $y_j$ from $(\vc y_t,t)$ is $\mata_j$ and from \Cref{matrix1}, the transition matrix of $z_t$ from $(\vc y_t\mata_j, t)$ is $\matb_t$.
\end{clamproof}

\begin{clam}
$n_\ell\leq|\Autom|+|\Butom|$.
\label{lastcounterbound}
\end{clam}
\begin{clamproof}
Assume for contradiction that $n_\ell>|\Autom|+|\Butom|$. For $i\in[1,n_\ell]$, let $(\vc z_i,i)$ denote the configuration pair where counter value $i$ is encountered the last time during the run of $w$. 
Let $w=x_iz_i$ where  $c_{\iota}\xrightarrow{x_i}(\vc z_i,i)\xrightarrow{z_i}c_{\ell}$
 and $\matc_i$ be such that $\vc z_{i} \matc_i= \vc x_{\ell}$. 
From the fundamental theorem of vector spaces, there exists 
$t\leq (|\Autom|+|\Butom|)+1$ such that $\vc z_t$ is a linear combination of $\vc z_1, \ldots, \vc z_{t-1}$. Since $\vc z_t \matc_t \trans{\vc \eta} \neq 0$, there exists $j<t$ such that $\vc z_{j} \matc_t \trans{\vc \eta} \neq 0$.  From \Cref{matrix2}, the transition matrix of $z_t$ from $(\vc z_j,j)$ is $\matc_t$. 
The word $x_jz_t$ contradicts the minimality of $w$. 
\end{clamproof}
Let $\K\in\N$ such that $|\Autom|,|\Butom|<\K$.
From Claim \ref{heightboundVoca} and \Cref{lastcounterbound}, we get that $"height"(w)\leq 2(\K+\K^2)$. Similar to that in \Cref{countersync}, we reduce the problem of finding the distinguishing word to equivalence check of \dfas of size $\mathcal{O}(\K^3)$.
\end{proof}

\section{Learning DROCAs}
\label{LearnDroca}

In this section, we give an $\Lstar$-like algorithm for active learning of "\drocas" using polynomially many queries with the help of a SAT solver. In this framework, we have a $learner$ and a $teacher$. The learner aims to construct a "\droca" that recognises the same language as the teacher's "\droca" (call it $\Autom$). The teacher answers the following types of queries by the learner.
\begin{itemize}
\item \emph{""membership queries""} \textsf{"MQ"}$_\Autom$: the learner provides a  word $w\in\Sigma^*$. The teacher returns $1$ if  $w\in"\Lang"(\Autom)$, and $0$ if $w\not\in"\Lang"(\Autom)$.
\item \emph{""counter value queries""} \textsf{"CV"}$_\Autom$: the learner asks  the counter value reached on reading word $w$. The teacher returns the counter value. i.e., $\ce_{\Autom}(w)$.
\item \emph{""minimal synchronous-equivalence queries""} \textsf{"MSQ"}$_\Autom$: the learner asks whether a "\droca" $\Cutom$ is "equivalent" {and "counter-synchronous"} to $\Autom$. The teacher returns \emph{yes} if $\Cutom$ and $\Autom$ are "counter-synchronous" and "equivalent". Otherwise, the teacher provides a minimal counter-example {$z \in \Sigma^*$} such that $"\Cutom(z)" \neq "\Autom(z)"$ or $\ce_{\Autom}(z) \neq \ce_{\Cutom}(z)$. 
\end{itemize}

The requirement for the teacher to return a counter-example $z$ such that \final{the counter values reached on reading $z$ in the teachers and learners "\drocas" are different} can be removed. But in that case, we have to use the "equivalence" check for "\drocas" (not necessarily "counter-synchronous").
{The number of queries needed will be polynomial in the length of the minimal counter-example that distinguishes two "\drocas". }
 Including this additional condition allows the use of the "counter-synchronous" "equivalence" check (see \Cref{countersync}), which is much faster.
The assumption of the teacher returning a minimal counter-example can be removed if we use "partial-equivalence queries". A ""partial-equivalence query"" takes two "\drocas" and a limit as inputs and determines whether the two "\drocas" are "equivalent" for all words whose length does not exceed the limit. If they are not, then it returns a distinguishing word whose length is less than the limit.
Neider and L{\"{o}}ding~\cite{christof} and Bruyère et al.~\cite{gaetan} used "partial-equivalence queries" to learn the behaviour of "\drocas" up to a certain counter value. 
One can simulate the teacher returning a minimal counter-example by a polynomial number of "partial-equivalence queries". 
Similarly, a "partial-equivalence query" can be replaced with an equivalence query that returns a minimal counter-example. 

\begin{table}[]
\centering\scalebox{1}{
\begin{tabular}{l|c|c|cc|cc|}
                     & \multirow{2}{*}{} & \multirow{2}{*}{\ce} & \multicolumn{2}{c|}{$\epsilon$}       &\multicolumn{2}{c|}{$a$}                    \\ \cline{4-7} 
                     &                   &                                & \multicolumn{1}{c|}{\Memb} & $\Actions$  & \multicolumn{1}{c|}{\Memb} & $\Actions$     \\ \Xhline{1pt} 
\multirow{5}{*}{\rotatebox{90}{$\P$}}   & $\epsilon$        & 0                              & \multicolumn{1}{c|}{0}          & $(0,+1,+1)$ &  \multicolumn{1}{c|}{0} & $(1,+1,-1)$\\ \cline{2-7} 
                     & a                 & 1                              & \multicolumn{1}{c|}{0}          & $(1,+1,-1)$ &  \multicolumn{1}{c|}{0} & $(1,+1,-1)$\\ \cline{2-7} 
                     & ab                & 0                              & \multicolumn{1}{c|}{0}          & $(0,0,+1)$ &  \multicolumn{1}{c|}{1} & $(0,+1,+1)$\\ \cline{2-7} 
                     & aba               & 0                              & \multicolumn{1}{c|}{1}          & $(0,+1,+1)$ &  \multicolumn{1}{c|}{0} & $(1,+1,+1)$\\
                     \cline{2-7} 
                     & b                 & 1                              & \multicolumn{1}{c|}{0}          & $(1,+1,+1)$ &  \multicolumn{1}{c|}{0} & $(1,+1,+1)$\\ \cline{2-7} 

                      \Xhline{1pt}
\multirow{5}{*}{\rotatebox{90}{$\P\Sigma$}}                      & aa                & 2                              & \multicolumn{1}{c|}{0}          & $(1,+1,-1)$ &  \multicolumn{1}{c|}{0} &$(1,+1,-1)$\\ \cline{2-7} 
                     & abb               & 1                              & \multicolumn{1}{c|}{0}          & $(1,+1,+1)$ &  \multicolumn{1}{c|}{0} &$(1,+1,+1)$\\ \cline{2-7} 
                     & abaa              & 1                              & \multicolumn{1}{c|}{0}          & $(1,+1,+1)$ &  \multicolumn{1}{c|}{0} &$(1,+1,+1)$\\ \cline{2-7} 
                     & abab              & 1                              & \multicolumn{1}{c|}{0}          & $(1,+1,+1)$ &  \multicolumn{1}{c|}{0} &$(1,+1,+1)$\\ \cline{2-7} 
                     & ba                 & 2                              & \multicolumn{1}{c|}{0}          & $(1,+1,+1)$ &  \multicolumn{1}{c|}{0} & $(1,+1,+1)$\\ \cline{2-7} 
		& bb                 & 2                              & \multicolumn{1}{c|}{0}          & $(1,+1,+1)$ &  \multicolumn{1}{c|}{0} & $(1,+1,+1)$\\ 
                     \hline
\end{tabular}
}
\caption{An "observation table" of the "\droca" given in \Cref{vocaEx} recognising the language $\{a^nb^na \mid n>0\}$. Here, $\P=\{\epsilon, a, ab,aba\}$ and $\S=\{\epsilon,a\}$.} 
\label{obtable}
\end{table}
\subsection{Observation Table}
Our algorithm maintains an ""observation table"" $C$ over the input alphabet $\Sigma=\{\sigma_1,\sigma_2,\ldots, \sigma_k\}$ for some $k\in\N$. 
$C=(\P,\S,\Memb, "\CV\upharpoonright_{\P\cup\P\Sigma}@\CV\upharpoonright", \Actions)$, where 
$\intro[\P]{}\P \subseteq \Sigma^*$ is a nonempty prefix-closed set of strings, $\intro[\S]{}\S\subseteq \Sigma^*$ is a nonempty suffix-closed set of strings,  $\intro[\Memb]{}\Memb: (\P \cup \P \Sigma)  \S \to \{0, 1\}$ \final{is a function that indicates whether words belong to the language}, $""\CV\upharpoonright_{\P\cup\P\Sigma}@\CV\upharpoonright"": \P\cup\P\Sigma \to \N$ is the function $\ce$ with domain restricted to the set $\P\cup\P\Sigma$, and $\Actions: (\P \cup \P \Sigma)  \S \to \{0,1\}\times\{0,1,-1\}^k$ \final{is a function representing the sign of the counter value reached and the counter-actions on every letter after reading a word}. 
Given $w\in (\P\cup \P  \Sigma)  \S$, $\Memb(w)$ is equal to $0$ (resp. $1$) if $"\Autom(w)"$ is equal to $0$ (resp. $1$), and $\intro[\Actions]{}\Actions(w)= ("\sgn"(\ce(w)), \ce(w \sigma_1)-\ce(w), \ldots, \ce(w \sigma_k)-\ce(w))$. 
Given $w_1,w_2\in(\P \cup \P \Sigma)  \S$, we say that $\Actions(w_1)$ is not similar to $\Actions(w_2)$ if and only if $"\sgn"(\ce(w_1))= "\sgn"(\ce(w_2))$ and $\Actions(w_1)\neq\Actions(w_2)$. {We use $\Actions(w_1)""\not\sim""\Actions(w_2)$ to denote this.} 
The "observation table" initially has $\P=\S=\{\epsilon\}$ and is augmented as the algorithm runs.

An "observation table" can be viewed as a two-dimensional array with rows labelled with elements of $\P \cup \P \Sigma$, columns labelled by elements of $S$ and an additional column labelled $\ce$. The column $\ce$ contains the counter value reached on reading the word labelling a row (see \Cref{obtable}). 
For any $p \in \P \cup \P \Sigma$ and $s\in \S$, the entry in row $p$ and column $s$ is equal to $(\Memb(p s), \Actions(p s))$ and $\cell{p}$ denotes the finite function $f_p$ from $\S$ to $\{0,1\} \times \{0,1\}\times\{0,1,-1\}^k$ 
 defined by $f_p(s)= (\Memb(p  s), \Actions(p s))$.  
 \AP 
We use $""row""(p)$ to denote $(\CV(p), \cell{p})$. For $p,p^\prime \in \P \cup \P \Sigma$, we say $"row"(p)$ is equal to $"row"(p^\prime)$ (denoted by $"row"(p)= "row"(p^\prime)$), if $\cell{p}=\cell{p^\prime}$ and $\CV(p)=\CV(p^\prime)$. 

Now, we introduce the notion of "$d$-closed" and "$d$-consistent" "observation tables" for any $d\in\N$. This is similar to the notion of closed and consistency used by Angluin~\cite{angluin}, but it also takes into account the counter values. 
\begin{definition}
Let $(\P,\S,\Memb,"\CV\upharpoonright_{\P\cup\P\Sigma}@\CV\upharpoonright",Actions)$ be an "observation table" and $d\in\N$. 
\begin{enumerate}
\item The "observation table" is not ""$d$-closed"" if there exist $p\in \P$ and $a\in\Sigma$ such that $\CV(p a)\leq d$ and for all $p^\prime \in \P$, $"row"(p a) \neq "row"(p^\prime)$. 
\item The "observation table" is not ""$d$-consistent"" if there exist $p,q\in \P$ and $a\in\Sigma$ such that $\CV(p)=\CV(q) \leq d$, $"row"(p)= "row"(q)$, and $"row"(p a)\neq "row"(q a)$. 
\end{enumerate}
\end{definition}

Consider the "observation table" given in \Cref{obtable}. This table is $1$-closed but not $2$-closed. This is because of the presence of the words $aa, ba$ and $bb$ in $\P\Sigma$. 
The given table is trivially $1$-consistent as there are no equal rows in $\P$.

\subsection{Constructing a DROCA from an Observation Table}
We introduce the notion of a "\encodedDFA". 
Given an alphabet $\Sigma$, we define the modified alphabet $""\widetilde\Sigma"" = \bigcup_{a\in\Sigma}\{ a^0, a^1\}$.
For a "\droca" \Autom, we define a function $\intro[\Enc]{}\Enc_{\Autom}: \Sigma^* \to "\widetilde\Sigma"^*$ as follows:
For $w\in \Sigma^+$, $\Enc_{\Autom}(w)= \tilde w$, such that for all $i\in [0,|w|-1]$, $\tilde w[i] = w[i]^{"\sgn"(\ce_{\Autom}({w_{[0\cdots i-1]}}))}$. Also $\Enc_{\Autom}(\epsilon)= \epsilon$. 
\begin{definition}
\label{chardfa}
Let $\Autom= (Q, \Sigma, q_0, \delta_0, \delta_1, F)$ be a "\droca".
The ""\encodedDFA"" $\Dutom_{\Autom}$ of $\Autom$ over the modified alphabet $"\widetilde\Sigma"$ is $\Dutom_{\Autom}= (Q, "\widetilde\Sigma", q_0, \delta, F)$ where, for all $q\in Q$ and $a\in\Sigma$,  
$\delta(q, a^0) =p$ (resp. $\delta(q, a^1)=p$) if and only if  $\delta_0(q, a)= (p,c)$ (resp. $\delta_1(q, a)=(p,c)$) for some $p\in Q$ and $c\in\{0,1,-1\}$.
\end{definition}
\Cref{underlyingDFA} shows the "\encodedDFA" over the modified alphabet $"\widetilde\Sigma"$ corresponding to the "\droca"  given in \Cref{vocaEx}. Note that in a "\encodedDFA", the counter information is not present. If we have access to the counter-actions, then we can construct a "\droca" from this "\encodedDFA". 

\begin{figure}[!h]
  {
   \begin{subfigure}[b]{0.45\textwidth}   
  \centering\scalebox{1}{
\begin{tikzpicture}[shorten >=1pt,node distance=3cm,on grid,auto]
\tikzset{every path/.style={line width=.4mm}}\
\node[state] at (-1,0) (q_0) {$q_0$};
\node[state] at (3.5,0) (q_1) {$q_1$};
\node[state,accepting] at (3.5,-2.5)(q_2){$q_2$};
\node[state] at (-1,-2.5)(q_3){$q_3$};
\path[->]
(q_0) edge [loop above] node {$a_{=0}/{\small+1},\ a_{>0}/{\small+}1$} ()
edge [below] node[xshift=-1.1cm,yshift=.3cm]  {$b_{=0}/{\small+1}$} (q_3)
edge [above] node[xshift=.1cm,yshift=.2cm] {$b_{>0}/{\small-1}$} (q_1)
(q_1) edge [loop above] node {$b_{>0}/{\small-1}$} ()
edge [above] node[xshift=.9cm, yshift=-.2cm] {$a_{=0}/{\small0}$} (q_2)
edge [above] node[xshift=-.3cm, yshift=0cm] {$a_{>0}/{\small+1}$} (q_3)
(q_2) edge [above] node[xshift=.2cm]{$a_{=0}/{\small+1},\ a_{>0}/{\small+}1$} (q_3)
(q_3) edge [loop below] node {$a_{=0}/{\small+1},\ a_{>0}/{\small+}1$} ();
\draw (1.5,-2.8) node {$b_{=0}/{\small+1},\ b_{>0}/{\small+}1$};
\draw (2.2,-1.4) node {$b_{=0}/{\small+}1$};
\draw (-1,-4.2) node {$b_{=0}/{\small+1},\ b_{>0}/{\small+}1$};
\draw[<-] (-1.53,0)-- (-2.05,0);
\end{tikzpicture}}
\caption{\centering\droca.}
\label{vocaEx}
\end{subfigure}
}
\hfill
  {
   \begin{subfigure}[b]{0.45\textwidth}   
  \centering\scalebox{1}{
\begin{tikzpicture}[shorten >=1pt,node distance=3cm,on grid,auto]
\tikzset{every path/.style={line width=.4mm}}\

\node[state] at (-1,0) (q_0) {$q_0$};
\node[state] at (3.5,0) (q_1) {$q_1$};
\node[state,accepting] at (3.5,-2.5)(q_2){$q_2$};
\node[state] at (-1,-2.5)(q_3){$q_3$};
\path[->]
(q_0) edge [loop above] node {$a^0, a^1$} ()
edge [below] node[xshift=-.2cm,yshift=.3cm] {$b^0$} (q_3)
edge [above] node[xshift=.1cm,yshift=.2cm] {$b^1$} (q_1)
(q_1) edge [loop above] node {$b^1$} ()
edge [below] node[xshift=.3cm, yshift=.3cm] {$a^0$} (q_2)
edge [below]  node[xshift=1cm, yshift=.3cm] {$a^1,b^0$} (q_3)
(q_2) edge [below] node[xshift=.1cm,yshift=0cm] {$a^0,a^1,b^0, b^1$} (q_3)
(q_3) edge [loop below] node {$a^0, a^1,b^0,b^1$} (); 
\draw[<-] (-1.53,0)-- (-2.05,0);
\draw[gray!5] [-] (0,-4.5)-- (0,-4.5);
\end{tikzpicture}}
\caption{\centering Characteristic \dfa.}
\label{underlyingDFA}
\end{subfigure}
}
\caption{A "\droca" recognising the language $\{a^nb^na \mid n>0\}$ is given in \Cref{vocaEx}. The "\encodedDFA" corresponding to this "\droca" is shown in \Cref{underlyingDFA}.} 
\label{DrocaDfa}
\end{figure}
\subsubsection{Constructing a Characteristic DFA Using a SAT Solver:}
\label{constructAutomaton}

\begin{algorithm}
\SetKwFunction{proc}{$\texttt{ConstructAutomaton}$}{}{}%
\SetKwProg{myproc}{Procedure}{}{end}
\SetKwInOut{Input}{Input}
\SetKwInOut{Output}{Output}
\SetKwFunction{connect}{Connect}
\SetKwData{ce}{ce}
\myproc{\proc{}}{
    \Input{Observation table $C=(\P,\S,\Memb,\CV\upharpoonright_{\P\cup\P\Sigma},\Actions)$}
    \Output{\encodedDFA $\Dutom_C$}
    \BlankLine
    Initialise $Pos=\emptyset,\ Neg=\emptyset$, $Operations=\{Actions(w)\mid w=p.s \text{ for some } p\in\P \text{ and }s\in\S\}$.\\
        \ForEach{$p$ in $\P$}{
        \ForEach{$s$ in $\S$}{
               \textbf{if} $\Memb(p s)=1$  \textbf{then} add $\Enc(p s)$ to $Pos$.\\
             \textbf{else}   
                add $Enc(p s)$ to $Neg$.\\
               add $\Enc(p s)\cdot \Actions(p s)$ to $Pos$.\\
             \ForEach{$op \in Operations$}{
            	\textbf{if} $\Actions(p s)\not\sim op$ \textbf{then} add $\Enc(p s)\cdot op$ to $Neg$.\\
		
            }

        }
    }
   $\Butom$= \texttt{find\_dfa}$(Pos,Neg)$.\\
   Remove transitions on $Operations$ from $\Butom$ to obtain a \dfa  $\Dutom_C$.\\
   \textbf{return}  $\Dutom_C$.
}
\caption{Algorithm to construct a \dfa from an observation table.}
\label{constructDFA}
\end{algorithm}

Let $C=(\P,\S,\Memb,"\CV\upharpoonright_{\P\cup\P\Sigma}@\CV\upharpoonright",\Actions)$ be an "observation table". 
We give $C$ as input to a procedure $\texttt{ConstructAutomaton}$ (see \Cref{constructDFA}) and obtain a \dfa $\Dutom_C$ over the modified alphabet $"\widetilde\Sigma"$ as output. The \dfa $\Dutom_C$ satisfies the following two properties: 
\begin{enumerate}
\item \label{cond1}for all $p\in\P\cup\P\Sigma$ and $s\in\S$, $\Enc_{\Autom}(ps)\in"\Lang"(\Dutom_C)$ if and only if $\Memb(ps)=1$.
\item \label{cond2}for any $p_1,p_2\in\P\cup\P\Sigma$ and $s_1,s_2\in\S$, if the run on $\Enc_{\Autom}(p_1s_1)$ and $\Enc_{\Autom}(p_2s_2)$ reaches the same state in $\Dutom_C$ then $\Actions(p_1s_1)$ is similar to $\Actions(p_2s_2)$. 
\end{enumerate}

We outline the operations performed by the function $\texttt{ConstructAutomaton}$. 
First, we create two sets of words, $Pos$ and $Neg$. For any $p\in\P\cup\P\Sigma$ and $s\in\S$, we add $\Enc_{\Autom}(ps)$ to $Pos$ (resp. $Neg$) if and only if $\Memb(ps)=1$ (resp. $0$). 
The \dfa $\Dutom_C$ will accept all words in $Pos$ and reject all words in $Neg$. 
This ensures condition \ref{cond1}. Observe that not all words over $"\widetilde\Sigma"$ correspond to encodings of words over $\Sigma$. 
To ensure condition \ref{cond2}, we add words over a larger alphabet $"\widetilde\Sigma" \cup "Operations"$, 
\final{where $""Operations""=\{Actions(w)\mid w=ps \text{ for some } p\in\P \text{ and }s\in\S\}$. 
For any $p\in\P\cup\P\Sigma$ and $s\in\S$, we add $\Enc_{\Autom}(p s)\cdot \Actions(p s)$ to $Pos$ and for all $op\in "Operations"$ where $\Actions(p s)"\not\sim" op$, we add $\Enc_{\Autom}(p s)\cdot op$ to $Neg$.} 
We find the "minimal separating \dfa" for the sets $Pos$ and $Neg$ using the function \texttt{find\_dfa}. We remove the transitions labelled by the letters from $"Operations"$ in this "minimal separating \dfa" to obtain $\Dutom_C$. 

Given an "observation table", the sets $Pos$ and $Neg$ can be constructed in polynomial time. {Every state of $\Autom$ can add two elements to the set $"Operations"$ -- one corresponding to the counter-actions on reading letters from counter value zero and the other for reading letters from positive counter value. Hence, the cardinality of this set is at most $2|\Autom|$.} 
\final{In our implementation, we use the algorithm by Dell'Erba et al.~\cite{dfaminer} that uses a SAT solver to find a minimal separating \dfa.}
However, the function \texttt{find\_dfa} can be any algorithm {that finds a minimal separating \dfa}~\cite{heule,Leucker,neidersat}. 
In the next subsection, we observe that $\Dutom_C$ is a "\encodedDFA".

\subsubsection{Constructing a DROCA from a Characteristic DFA: }
Let the input alphabet be $\Sigma = \{\sigma_1,\ldots,\sigma_k\}$ and $C=(\P, \S, \Memb,"\CV\upharpoonright_{\P\cup\P\Sigma}@\CV\upharpoonright",\Actions)$ be an "observation table" over $\Sigma$. The following lemma states that we can construct a "\droca" $\Butom_C$ from $C$ such that it agrees with the "observation table" $C$.  The idea is to use the "observation table" to assign counter-actions to transitions of $\Dutom_C$. 
\begin{lemma}
\label{lem:constructDROCA}
Given an "observation table" $C$, {we can construct} a "\droca" $\Butom_C$ with $|\Butom_C|\leq|\Autom|$ such that for all $p\in\P\cup\P\Sigma$ and $s\in\S$, $\Butom_C(ps)=\Memb(ps)$
and $\ce_{\Butom_C}(ps)=\ce_{\Autom}(ps)$. 
\end{lemma}
\begin{proof}
\label{lemmaSev}
Let the function \upshape$\texttt{ConstructAutomaton}$  for input $C$ return a \dfa $\Dutom_C= (Q, "\widetilde\Sigma", q_0, \delta, F)$. 
For all $p\in\P\cup\P\Sigma$ and $s\in\S\cup\S\Sigma$, we can find $\ce(ps)$ from $C$. 
We define the "\droca" $\Butom_C= (Q, \Sigma, q_0, \delta_0, \delta_1, F)$ where $\delta_0$ and $\delta_1$ are specified as follows.
For all $q\in Q$, $a\in\Sigma$,  $\delta_0(q, a)= (\delta(q, a^0),c)$ for some $c\in\{0,1\}$, if there exists $p\in\P\cup\P\Sigma$ and $s\in\S$ such that $\Dutom_C$ on reading ${\Enc_{\Autom}(ps)}$ reaches the state $q$ with $\ce(ps)=0$ and $\ce(psa)=c$.
Similarly, for all $q\in Q$, $a\in\Sigma$, $\delta_1(q, a)= (\delta(q, a^1),c)$ for some $c\in\{0,1,-1\}$, if there exists $p\in\P\cup\P\Sigma$ and $s\in\S$ such that $\Dutom_C$ on reading ${\Enc_{\Autom}(ps)}$ reaches the state $q$ with $\ce(ps)>0$ and $\ce(psa)=c$. 
\final{If there are any transitions that do not have a counter-action, it means that there is no word in the "observation table" that took that transition. One can assign any arbitrary counter-action for these transitions.}

We know that for any word $w, w^\prime\in\Sigma^*$, if $\Dutom_C$ on reading $w$ and $w^\prime$ reaches the same state, then $\Actions(w)$ is similar to $\Actions(w^\prime)$. In our construction, we are assigning counter-actions to transitions from a state in $\Butom_C$ based on the counter-actions of words that reach that state in $\Dutom_C$. Since all words that reach a state have similar counter-actions, this assignment will be consistent.  
By construction, for all $p\in\P\cup\P\Sigma$ and $s\in\S$, $\ce_{\Butom_C}(ps)= \ce_{\Autom}(ps)$. This also ensures that $\Enc_{\Autom}(ps)= \Enc_{\Butom_C}(ps)$. Note that, $\Dutom_C$ on reading $\Enc_{\Autom}(ps)$ reaches a final state if and only if $ps\in"\Lang"(\Autom)$. Since $\Enc_{\Autom}(ps)= \Enc_{\Butom_C}(ps)$, we get that $\Butom_C$ reaches a final state on reading $ps$ if and only if $ps\in"\Lang"(\Autom)$. Therefore, for all $p\in\P\cup\P\Sigma$ and $s\in\S$, $\Butom_C(ps)=1$ if and only if $\Memb(ps)=1$.
\end{proof}

\subsection{\minOCA: The Learning Algorithm}
\Cref{algdroca} learns a "\droca" that is "equivalent" to an input "\droca" $\Autom$.
Initially, it sets up an "observation table" using empty strings and incrementally refines this table to distinguish states of the unknown "\droca". The process iteratively increases an integer value $d$ and uses "membership" and "counter value queries" to construct a "$d$-closed" and "$d$-consistent" "observation table"  $C=(\P, \S, \Memb,"\CV\upharpoonright_{\P\cup\P\Sigma}@\CV\upharpoonright",\Actions)$ \final{(see lines \ref{whileloop} -- \ref{whileEnd} of \Cref{algdroca}). This part resembles the $\Lstar$ algorithm.} Making the table "$d$-closed" adds new rows to $\P$, and making it "$d$-consistent" adds new columns to $\S$. 
We construct a "\droca" from a "$d$-closed" and "$d$-consistent" "observation table" $C$ using \Cref{lem:constructDROCA} and ask a "minimal synchronous-equivalence query". If the teacher provides a counter-example, then all its prefixes are added to $\P$, and the value of $d$ is updated to the "height" of the counter-example, if it is more than $d$.
The table is then extended until it becomes "$d$-closed" and "$d$-consistent". This process continues until the correct "\droca" is learnt.

\begin{algorithm}[h]
\LinesNumbered
\caption{ ""\text{\upshape\minOCA}@\minOCA"": "\droca" learning algorithm.}\label{algdroca}
\SetKwInOut{Input}{Require}
\SetKwInOut{Output}{Ensure}
\Input{The teacher knowing a "\droca" $\Autom$.}
\Output{A "\droca" accepting the same language as $\Autom$ is returned.}
\BlankLine
Initialise $\P$ and $\S$ to $\{\epsilon\}$, and $d$ to $0$.\\
Initialise the "observation table" $C=(\P,\S,\Memb,"\CV\upharpoonright_{\P\cup\P\Sigma}@\CV\upharpoonright",\Actions)$.\\
\Repeat{teacher replies \emph{yes} to a "minimal synchronous-equivalence query"}
{
\While{$C$ is not "$d$-closed" or not "$d$-consistent" \label{whileloop}}{
	\If{$C$ is not "$d$-closed" }{
	Find $p\in \P$, $a\in\Sigma$ such that $\cv{p a}\leq d$, $"row"(p a) \neq "row"(p^\prime)$ for all $p^\prime \in \P$.\\ 
	Add $p a$ to $\P$.	\label{dclosed}}
\If{$C$ is not d-consistent}{
Find $p,q\in \P$, $a\in\Sigma$, $s\in \S$ such that $\cv{p}=\cv{q} \leq d$, $"row"(p)= "row"(q)$, and   
{($\Memb(p a  s) \neq \Memb(q a  s)$ or $\Actions(p  a  s) \neq \Actions(q  a  s)$).}\\
{
Add $a  s$ to $\S$. \label{dconsistent}}
}
Extend $\Memb$ and $\Actions$ to $(\P\cup \P \Sigma) \S$, using "membership" and "counter value queries".

}\label{whileEnd}

Construct a "\droca" $\Butom_C$ from $C$ using \Cref{lem:constructDROCA}. \label{constructdroca} \\
Ask "minimal synchronous-equivalence query" \textsf{"MSQ"}$_\Autom(\Butom_C)$. \label{minequivquery} \\
\If{teacher gives a counter-example $z$}{
Add $z$ and all its prefixes to $\P$ and extend $\Memb$ and $\Actions$ to $(\P\cup\P \Sigma) \S$ using "membership" and "counter value queries". \label{counterexample}\\
\If{$"height_{\Autom}"(z)> d$}{
	$d="height_{\Autom}"(z)$.
	}
} 
	
}
Halt and output $\Butom_C$.
\end{algorithm}

\subsubsection{Example: \minOCA in Action.}
We assume the learner is trying to learn the "\droca" given in \Cref{exminOCA1} from the teacher. The learner uses \Cref{algdroca} ("\minOCA") to learn this language. First, the learner initialises $"\P"$ and $"\S"$ to $\{\epsilon\}$ and $d$ to $1$. The initial "observation table" $C=("\P","\S","\Memb","\CV\upharpoonright_{\P\cup\P\Sigma}","\Actions")$ is built as shown in \Cref{tabExMinOCA1} using "membership" and "counter value queries". 

\begin{figure}[b]
\centering
\scalebox{1}{
\begin{tikzpicture}[shorten >=1pt,node distance=3cm,on grid,auto]
\tikzset{every path/.style={line width=.4mm}}\
   \tikzset{initial text={}}
   \node[state, initial] (q0)   at (0,0)   {$q_0$};
   \node[state,accepting]           (q1)   at (6,0)   {$q_1$};
   \node[state,accepting]           (q2)   at (3,0) {$q_2$};

   \path[->]
   (q0) edge[loop below] node {$b_{=0}/{\small0}$} (q0)
   (q1) edge[loop below] node {$a_{>0}/{\small}0$} (q1)
   (q0) edge[bend left] node[xshift=.2cm,yshift=.4cm]{$a_{=0}/{\small+1},\ a_{>0}/{\small}0$} (q2)
   (q2) edge[right] node[xshift=-.5cm,yshift=-.3cm] {$b_{>0}/{\small}0$} (q0)
   (q2) edge[bend left] node {$a_{>0}/{\small}0$} (q1)
   (q1) edge[left] node[xshift=.7cm, yshift=-.3cm] {$b_{>0}/{\small}0$} (q2);
   \draw (1.4,.9) node {$b_{>0}/0$};
\end{tikzpicture}
}
\caption[Example: The \droca under learning.]{Example: The "\droca" under learning.}
\label{exminOCA1}
\end{figure}

\begin{table}[h]
\centering
\begin{minipage}{0.32\textwidth}
\centering
\begin{tabular}{|c|c|c|c|}
 \multirow{2}{*}{} & \multirow{2}{*}{\ce}   & \multicolumn{2}{c|}{$\epsilon$}       \\ \cline{3-4} 
                &                          & \multicolumn{1}{c|}{Mem} & $Actions$\\ \Xhline{1pt}
$\epsilon$          & 0                      & 0          &$(0,+1,0)$    \\ \Xhline{1pt}
a          & 1                      & 1          &$(1,0,0)$     \\ \hline
b          & 0                      & 0          &$(0,+1,0)$    \\ \hline
\end{tabular}
\vspace{1.1cm}
\caption[Initial observation table with $\P=\S=\{\epsilon\}$.]{Initial "observation table" with $"\P"="\S"=\{\epsilon\}$.}
\label{tabExMinOCA1}
\end{minipage}
\hfill
\begin{minipage}{0.32\textwidth}
\centering
\begin{tabular}{|c|c|c|c|}
 \multirow{2}{*}{} & \multirow{2}{*}{\ce}   & \multicolumn{2}{c|}{$\epsilon$}       \\ \cline{3-4} 
                &                          & \multicolumn{1}{c|}{"\Memb"} & $"\Actions"$\\ \Xhline{1pt}
$\epsilon$          & 0                      & 0          &$(0,+1,0)$    \\ \hline
a          & 1                      & 1          &$(1,0,0)$     \\ \Xhline{1pt}
b          & 0                      & 0          &$(0,+1,0)$    \\ \hline
aa         & 1                      & 1          &$(1,0,0)$     \\ \hline
ab         & 1                      & 0          &$(1,0,0)$ \\  \hline

\end{tabular}
\vspace{.8cm}
\caption[An observation table with $\P=\{\epsilon, a\}$ and $\S=\{\epsilon\}$ that is not $1$-closed.]{An "observation table" with $"\P"=\{\epsilon, a\}$ and $"\S"=\{\epsilon\}$ that is not $1$-closed.}
\label{tabExMinOCA0}
\end{minipage}
\hfill
\begin{minipage}{0.32\textwidth}
\centering
\begin{tabular}{|c|c|c|c|}
 \multirow{2}{*}{} & \multirow{2}{*}{\ce}   & \multicolumn{2}{c|}{$\epsilon$}       \\ \cline{3-4} 
                &                          & \multicolumn{1}{c|}{"\Memb"} & $"\Actions"$\\ \Xhline{1pt}
$\epsilon$          & 0                      & 0          &$(0,+1,0)$    \\ \hline
a          & 1                      & 1          &$(1,0,0)$     \\ \hline
ab         & 1                      & 0          &$(1,0,0)$     \\ \Xhline{1pt}
b          & 0                      & 0          &$(0,+1,0)$    \\ \hline
aa         & 1                      & 1          &$(1,0,0)$     \\ \hline
aba        & 1                      & 1          &$(1,0,0)$     \\ \hline
abb        & 1                      & 1          &$(1,0,0)$     \\ \hline
\end{tabular}
\caption[A $1$-closed and $1$-consistent observation table with $\P=\{\epsilon, a,ab\}$ and $\S=\{\epsilon\}$.]{A $1$-closed and $1$-consistent "observation table" with $"\P"=\{\epsilon, a,ab\}$ and $"\S"=\{\epsilon\}$.}
\label{tabExMinOCA2}
\end{minipage}
\end{table}

The "observation table" shown in \Cref{tabExMinOCA1} is not "$d$-closed" since  $"\ce"(a)=d$ and $"row"(a)\neq "row"(p)$ for all $p\in"\P"$ with $"\ce"(p)= d$. Therefore, we add $a$ to $"\P"$ and fill the "observation table" using "membership" and "counter value queries" (see  \Cref{tabExMinOCA0}). 
However, this "observation table" is still not "$d$-closed" since $"\ce"(ab)=d$ and $"row"(ab)\neq "row"(p)$ for all $p\in"\P"$, with $"\ce"(p)=d$. Therefore, we add $ab$ to $"\P"$ and fill the "observation table" using "membership" and "counter value queries". 
The "observation table" after this step is shown in \Cref{tabExMinOCA2}. 
This "observation table" is both "$d$-closed" and "$d$-consistent". 
We use the procedure $\texttt{"ConstructAutomaton"}$ using this "observation table" as input to obtain a "\encodedDFA". 
Let $x$ denote the tuple $(0,+1,0)$ and $y$ denote the tuple $(1,0,0)$. The tuple $x$ (resp. $y$) represents the counter-actions on reading symbols from zero (resp positive) counter value. 
The sets $Pos$ and $Neg$ created inside the function $\texttt{"ConstructAutomaton"}$ are given below.
\begin{align*}
Pos&=\{x, a^0, a^0 y, b^0 x, a^0 a^1, a^0 a^1 y, a^0 b^1 y, a^0 b^1 a^1, a^0 b^1 a^1 y, a^0 b^1 b^1, a^0 b^1 b^1 y\}.\\
Neg&=\{\epsilon, b^0, a^0b^1\}.
\end{align*}

The "\encodedDFA" returned by $\texttt{"ConstructAutomaton"}$ is shown in \Cref{DFAexminOCA1}.

\begin{figure}[h]
\begin{minipage}{0.47\textwidth}
\centering
\begin{tikzpicture}[shorten >=1pt,node distance=3cm,on grid,auto]
\tikzset{every path/.style={line width=.4mm}}\
   \tikzset{initial text={}}
    \node[state, initial] (q0) at (0,0)  {$q_0$}; 
    \node[state, accepting] (q1) at (4,0) {$q_1$}; 

    \path[->] 
    (q0) edge [loop above] node {$b^0$} ()
         edge [bend right, below] node {$a^0,a^1,b^1,x,y$} (q1)
    (q1) edge [loop above] node {$a^1,y$} ()
    	edge [bend right, above] node {$a^0,b^0,b^1,x$} (q0)
    ;
\end{tikzpicture}
\caption[The \encodedDFA obtained from  \Cref{tabExMinOCA2}.]{The "\encodedDFA" obtained from $\texttt{"ConstructAutomaton"}$ using \Cref{tabExMinOCA2} as input.}
\label{DFAexminOCA1}
\end{minipage}
\hfill
\begin{minipage}{0.47\textwidth}
\centering
\begin{tikzpicture}[shorten >=1pt,node distance=3cm,on grid,auto]
\tikzset{every path/.style={line width=.4mm}}\
   \tikzset{initial text={}}
    \node[state, initial] (q0) at (0,0)  {$q_0$}; 
    \node[state, accepting] (q1) at (4,0) {$q_1$}; 

    \path[->] 
    (q0) edge [loop above] node {$b_{=0}/{\small0}$} ()
         edge [bend right, below] node {$a_{=0}/{\small+}1$, $a_{>0}/{\small0}$, $b_{>0}/{\small0}$} (q1)
    (q1) edge [loop above] node {$a_{>0}/{\small0}$} ()
    	edge [bend right, above] node {$b_{>0}/{\small0}$} (q0)
    ;
\end{tikzpicture}
\caption[The \droca obtained from the \encodedDFA in \Cref{DFAexminOCA1}.]{The "\droca" obtained from the "\encodedDFA" in \Cref{DFAexminOCA1}.}
\label{OCAexminOCA1}
\end{minipage}
\end{figure}

A "\droca" is constructed from this "\encodedDFA" by removing from it the transitions on $x$ and $y$, and by assigning counter-actions to the rest of the transitions based on the tuples $x$ and $y$. 
For instance, the tuple $x=(0,+1,0)$ is accepted from state $q_0$. This means that the counter-action on a transition on $a$ (resp. $b$) from counter value $0$ from this state has counter-action $+1$ (resp. $0$). 
The resultant "\droca" is shown in \Cref{OCAexminOCA1}.

The learner now asks a "minimal synchronous-equivalence query" with this "\droca" as input. This "\droca" is "counter-synchronous" with the teachers "\droca" but not equivalent. Therefore, the teacher returns a minimal counter-example. Let us assume that the counter-example returned by the teacher is $aab$.  We add $aab$ and all its prefixes to $"\P"$ and extend the "observation table" using "membership" and "counter value queries". The "observation table" obtained after this step is shown in \Cref{tabExMinOCA3}. Since $"height"(aab)=1=d$, we keep the value of $d$ unchanged.

\begin{table}[h]
\begin{minipage}{0.47\textwidth}
\centering
\begin{tabular}{|c|c|c|c|}
 \multirow{2}{*}{} & \multirow{2}{*}{\ce}   & \multicolumn{2}{c|}{$\epsilon$}       \\ \cline{3-4} 
                &                          & \multicolumn{1}{c|}{"\Memb"} & $"\Actions"$\\ \Xhline{1pt}
$\epsilon$          & 0                      & 0          &$(0,+1,0)$    \\ \hline
a          & 1                      & 1          &$(1,0,0)$     \\ \hline
ab         & 1                      & 0          &$(1,0,0)$     \\ \hline
aa         & 1                      & 1          &$(1,0,0)$     \\ \hline
aab        & 1                      & 1          &$(1,0,0)$     \\ \Xhline{1pt}
b          & 0                      & 0          &$(0,+1,0)$    \\ \hline
aba        & 1                      & 1          &$(1,0,0)$     \\ \hline
abb        & 1                      & 1          &$(1,0,0)$     \\ \hline
aaa        & 1                      & 1          &$(1,0,0)$     \\ \hline
aaba       & 1                      & 1          &$(1,0,0)$     \\ \hline
aabb       & 1                      & 0          &$(1,0,0)$     \\ \hline
\end{tabular}
\caption[An observation table with $\P=\{\epsilon, a, ab,aa, aab\}$ and $\S=\{\epsilon\}$ that is not $1$-closed.]{An "observation table" with $"\P"=\{\epsilon, a, ab,aa, aab\}$ and $"\S"=\{\epsilon\}$ that is not $1$-closed.}
\label{tabExMinOCA3}
\end{minipage}
\begin{minipage}{0.47\textwidth}
\centering
\begin{tabular}{|c|c|c|c|c|c|}
 \multirow{2}{*}{} & \multirow{2}{*}{\ce}   & \multicolumn{2}{c|}{$\epsilon$}    & \multicolumn{2}{c|}{$b$}    \\ \cline{3-6} 
                &                          & \multicolumn{1}{c|}{"\Memb"} & $"\Actions"$ & \multicolumn{1}{c|}{"\Memb"} & $"\Actions"$\\ \Xhline{1pt}
$\epsilon$  & 0  & 0     &     $(0,1,0)$ & 0 & $(0,+1,0)$ \\ \hline
a   & 1  & 1    &    $(1,0,0)$ & 0 & $(1,0,0)$ \\ \hline
ab  & 1  & 0   &     $(1,0,0)$ & 1 & $(1,0,0)$ \\ \hline
aa  & 1  & 1     &      $(1,0,0)$ & 1 & $(1,0,0)$ \\ \hline
aab & 1  & 1    &     $(1,0,0)$ & 0 & $(1,0,0)$ \\ \Xhline{1pt}
b   & 0  & 0     &      $(0,1,0)$ & 0 & $(0,+1,0)$ \\ \hline
aba & 1  & 1    &     $(1,0,0)$ & 0 & $(1,0,0)$ \\ \hline
abb & 1  & 1    &     $(1,0,0)$ & 0 & $(1,0,0)$ \\ \hline
aaa & 1  & 1     &     $(1,0,0)$ & 1 & $(1,0,0)$ \\ \hline
aaba & 1 & 1   &      $(1,0,0)$ & 1 & $(1,0,0)$ \\ \hline
aabb & 1 & 0   &     $(1,0,0)$ & 1 & $(1,0,0)$ \\ \hline
\end{tabular}
\caption[A $1$-closed and $1$-consistent observation table with $\P=\{\epsilon, a, ab,aa,aab\}$ and $\S=\{\epsilon,b\}$.]{A $1$-closed and $1$-consistent "observation table" with $"\P"=\{\epsilon, a, ab,aa,aab\}$ and $"\S"=\{\epsilon,b\}$.}
\label{tabExMinOCA4}
\end{minipage}
\end{table}

This "observation table" is not "$d$-consistent" since $"row"(aa)="row"(aab)$ but $aa\cdot  b \cdot \epsilon \neq aab \cdot b \cdot \epsilon$. Therefore, we add $b$ to $"\S"$ and extend the table using "membership" and "counter value queries". The "observation table" obtained after this step is shown in \Cref{tabExMinOCA4}. 
 This "observation table" is both "$d$-closed" and "$d$-consistent". Therefore, the learner uses the procedure $\texttt{"ConstructAutomaton"}$ using this "observation table" as input to obtain a "\encodedDFA". The sets $Pos$ and $Neg$ created during this process are given below.
\begin{align*}
Pos&=\{x, a^0, a^0y, b^0x, a^0a^1, a^0a^1y, a^0b^1y, a^0b^1a^1, a^0b^1a^1y, a^0b^1b^1, a^0b^1b^1y, \\
 &a^0a^1a^1, a^0a^1a^1y, a^0a^1b^1, a^0a^1b^1y, a^0a^1b^1a^1, a^0a^1b^1a^1y, a^0a^1b^1b^1y, b^0b^0x, \\
 & a0b1a1b1y, a0b1b1b1y, a0a1a1b1, a0a1a1b1y, a0a1b1a1b1, a0a1b1a1b1y, \\
 &a^0a^1b^1b^1b^1, a^0a^1b^1b^1b^1y\}.\\
Neg&=\{\epsilon, b^0, a^0b^1, a^0a^1b^1b^1, b^0b^0, a^0b^1a^1b^1, a^0b^1b^1b^1\}.
\end{align*} 

The "\encodedDFA" returned is shown in \Cref{DFAexminOCA2}.
A "\droca" is constructed from this "\encodedDFA" by removing from it the transitions on $x$ and $y$, and by assigning counter-actions to the remaining transitions based on the tuples $x$ and $y$. The resultant "\droca" is shown in  \Cref{OCAexminOCA2}.

\begin{figure}[!h]
\begin{minipage}{0.47\textwidth}
\centering
\begin{tikzpicture}[shorten >=1pt,node distance=3cm,on grid,auto]
\tikzset{every path/.style={line width=.4mm}}\
   \tikzset{initial text={}}
    \node[state, initial] (q0) at (0,0) {$q_0$}; 
    \node[state, accepting] (q1) at (4,0) {$q_1$}; 
    \node[state, accepting] (q2) at (2,-4) {$q_2$};
     
     \draw [gray!20] (2,-5.3)-- (2,-5.3) ;
    \path[->]
    (q0) edge [loop above] node {$b^0$} ()
         edge [bend left] node {$a^0,a^1,b^1,x,y$} (q1)
    (q1) edge [bend left] node {$a^0,a^1$} (q2)
         edge [loop right] node {$b^0,x,y$} ()
         edge [bend left] node {$b^1$} (q0)
    (q2) edge [bend left] node {$a^0$} (q0)
         edge [loop right] node {$a^1,b^0$} ()
         edge [bend left] node [xshift=-.2cm, yshift=-1cm]{$b^1,x,y$} (q1);
\end{tikzpicture}
\caption[The \encodedDFA obtained from \Cref{tabExMinOCA4}.]{The "\encodedDFA" obtained from $\texttt{"ConstructAutomaton"}$ using \Cref{tabExMinOCA4} as input.}
\label{DFAexminOCA2}
\end{minipage}
\hfill
\begin{minipage}{0.47\textwidth}
\centering
\begin{tikzpicture}[shorten >=1pt,node distance=3cm,on grid,auto]
\tikzset{every path/.style={line width=.4mm}}\
   \tikzset{initial text={}}
    \node[state, initial] (q0) at (0,0) {$q_0$}; 
    \node[state, accepting] (q1) at (5,0) {$q_1$}; 
    \node[state, accepting] (q2) at (2.5,-5) {$q_2$};

    \path[->]
    (q0) edge [loop above] node {$b_{=0}/{\small0}$} ()
         edge [bend left] node {$a_{=0}/{\small+}1$, $a_{>0}/{\small}0$, $b_{>0}/{\small}0$ } (q1)
    (q1) edge [bend left] node { $a_{=0}/{\small+}1, a_{>0}/{\small}0$} (q2)
         edge [bend left] node {$b_{>0}/{\small}0$} (q0)
           edge [loop right] node {$b_{=0}/{\small0}$} ()
    (q2) edge [loop right] node { $a_{>0}/{\small}0, b_{=0}/{\small0}$} ()
    	edge [bend left] node {$a_{=0}/{\small+}1$} (q0)
         edge [bend left] node [xshift=-.2cm, yshift=-1cm] {$b_{>0}/{\small}0$} (q1);
\end{tikzpicture}
\caption[The \droca obtained from the \encodedDFA in \Cref{DFAexminOCA2}.]{The "\droca" obtained from the "\encodedDFA" in \Cref{DFAexminOCA2}.}
\label{OCAexminOCA2}
\end{minipage}
\end{figure}
The learner now asks a "minimal synchronous-equivalence query" with this "\droca" as input. Since this is equivalent and "counter-synchronous" with the teachers "\droca", the teacher replies $yes$ to the "minimal synchronous-equivalence query". The algorithm halts by outputing the learnt "\droca" shown in \Cref{OCAexminOCA2}.

\subsubsection{Analysis of "\minOCA":}
The correctness of \Cref{algdroca} follows from the fact that whenever the teacher replies yes to a "minimal synchronous-equivalence query", the learnt "\droca" recognises the target language.
Now, we show that the algorithm terminates after polynomially many queries. 

\AP
Similar to \cite{christof}, we define a refined  Myhill-Nerode congruence $""\simeq"" \subseteq \Sigma^*\times\Sigma^*$ as follows: for  $u,v\in\Sigma^*$, $u"\simeq" v$ if and only if for all $z\in\Sigma^*$, $"\Autom(uz)"="\Autom(vz)"$ and  $\ce_{\Autom}(uz)=\ce_{\Autom}(vz)$. 
 { For all $w\in\Sigma^*$, let $[w]$ denote the equivalence class of $w$ under $"\simeq"$. The behaviour graph of a "\droca" is an infinite state automaton induced by this refined congruence $"\simeq"$. 
Each equivalence class under this refined congruence corresponds to a state of the behaviour graph, and the transitions between them are defined based on the equivalence class of the resultant words. i.e., for $u\in\Sigma^*$ and $a\in\Sigma$ there is a transition from the state corresponding to $[u]$ to the state corresponding to $[ua]$ on reading the symbol $a$. The state corresponding to an equivalence class $[u]$ will be marked as a final state in the behaviour graph if and only if $"\Autom(u)"=1$.  

Let $\parsimeq{C}$ denote the restriction of $"\simeq"$ to the entries in the "observation table" $C$. i.e., for all $p,p'\in\P\cup\P\Sigma$ and $s,s'\in\S$, $ps\parsimeq{C}\  p's'$ if and only if for all $z\in\Sigma^*$ where both $psz$  and $p's'z$ are in $(\P\cup\P\Sigma)\S$, $"\Autom(psz)@\Lang"="\Autom(p's'z)@\Lang"$ and  $\ce_{\Autom}(psz)=\ce_{\Autom}(p's'z)$. 
Given an "observation table" $C$, we can construct a partial behavioural graph $BG_{C}$ induced by $\parsimeq{C}$ in a similar fashion. 
Note that for all $p\in\P\cup\P\Sigma$ and $s\in\S$, the counter value corresponding to the equivalence class reached on reading $ps$ and the membership of $ps$ in $BG_C$ is the same as the counter value reached and membership of $ps$ in $C$.  
 }
\begin{proposition}
\label{sameRow}
Given $p,p'\in\P$, if $p "\simeq" p^\prime$, then $"row"(p)="row"(p')$. 
\end{proposition}

In the next proposition, $C$ is an "observation table" and $\Butom_C$ the $\droca$ learnt by the learner from $C$ (in Line \ref{constructdroca} of \Cref{algdroca}).
\begin{proposition} 
\label{propos1}
Let $z$ be a counter-example, with $"height_{\Autom}"(z)= d$, returned by {\upshape\textsf{"MSQ"}}$_{\Autom}(\Butom_C)$. Let $C'$ be the "observation table" obtained by adding the prefixes of $z$ to $\P$ in $C$ and making it "$d$-closed" and "$d$-consistent". Then the number of distinct rows with counter value less than or equal to $d$ in $C'$ is more than that of $C$.
\end{proposition}
\begin{proofsketch}
Let $z = ${\upshape\textsf{"MSQ"}}$_{\Autom}(\Butom_C)$ be a counter-example with $"height_{\Autom}"(z)= d$. 
{Let $C'$ be the "observation table" obtained by adding the prefixes of $z$ to $\P$ in $C$ and making it "$d$-closed" and "$d$-consistent".  Assume for contradiction that the number of distinct rows with counter value less than or equal to $d$ in $C'$ is the same as that of $C$. From \Cref{sameRow}, we get that $BG_C$ and $BG_{C'}$ are the same in this case. From \Cref{lem:constructDROCA}
for all $p\in\P\cup\P\Sigma$ and $s\in\S$, the counter values reached and membership of $ps$ are the same in $\Butom_C$ and $C$. We also know that the counter value corresponding to the equivalence class reached on reading $ps$ and membership of $ps$ in $BG_C$ is the same as the counter value reached and membership of $ps$ in $C$.  Since $BG_{C}= BG_{C'}$, the membership and counter values reached by all prefixes of $z$ in $\Autom$ should also match $\Butom_C$ contradicting the assumption that $z$ is a counter-example.}
\end{proofsketch}

\begin{proposition}
\label{propos2}
For any $d \in \N$, at most $d \times |\Autom|$ many counter-examples of "height" less than or equal to $d$ is returned by the "minimal synchronous-equivalence query".
 \end{proposition}
\begin{proofsketch}
Fix a $d' \in \N$. There are at most $|\Autom|$ many configurations of $\Autom$ with counter value $d'$.
We know that for $p,p'\in\P$, if $"row"(p) \neq "row"(p')$, then $p "\not\simeq" p'$. 
Hence, there are at most $|\Autom|$ distinct rows in the "observation table" with counter value $d'$. 
Consequently, there are at most $d \times |\Autom|$ distinct rows with counter values $d' \leq d$. The claim now follows from \Cref{propos1}.
\end{proofsketch}

\begin{proposition}
\label{propos3}
 At most $|\Autom|^5+1$ many "minimal synchronous-equivalence queries" are executed during the run of \Cref{algdroca}.
\end{proposition}
\begin{proof}
 Assume for contradiction that more than $|\Autom|^5+1$ "minimal synchronous-equivalence queries" are executed. Hence, more than $|\Autom|^5$ counter-examples are returned by these queries. From \Cref{propos2}, it follows that at least one counter-example is of "height" greater than $|\Autom|^4$. However, from \Cref{countersync}, we know that the "height" of the minimal counter-example that distinguishes two "\drocas" of size $|\Autom|$ is at most $|\Autom|^4$. This is a contradiction, since by \Cref{lem:constructDROCA}, $|\Butom_C| \leq |\Autom|$ for any "observation table" $C$ during the run of the algorithm.
\end{proof}

\begin{theorem}
Given a "\droca" $\Autom$, a minimal "counter-synchronous" "\droca" recognising the same language can be learnt with at most $|\Autom|^5+1$ queries to the SAT solver, $|\Autom|^5+1$ "minimal synchronous-equivalence queries" and polynomially many "membership" and "counter value queries". 
\end{theorem}

We can adapt the learning algorithm for "\drocas" for learning "\vocas". A minimal "\voca" can be learnt using at most $2|\Autom|^3+1$ equivalence queries from \Cref{vocaeq}.
For "\vocas", we do not require the "counter value queries". 
 {Also, there are words that do not have a valid run. i.e., the counter goes below zero during the run. These words will not have a corresponding entry in the "observation table" and will be treated as don't care. In this case, \Cref{algdroca} can be simplified further as we don't have to create the set $"Operations"$ to encode the counter actions, as the actions on the counter are already determined by the input alphabet.}
\begin{corollary}
\label{LVoca}
Given a "\voca" $\Autom$, a minimal "\voca" recognising the same language can be learnt using at most $\mathcal{O}(|\Autom|^3)$ queries to the SAT solver, $\mathcal{O}(|\Autom|^3)$ "minimal synchronous-equivalence queries" and polynomially many "membership queries". 
\end{corollary}

\section{Implementation}
\label{sec:implementation}
The proposed method, henceforth denoted as "\minOCA", was implemented in Python\footnote{The implementation of "\minOCA", the datasets used, and the complete results generated can be found in the following link: https://doi.org/10.5281/zenodo.14604419 \cite{minOCAimp}.} and is tested against a set of randomly generated "\drocas". In this section, we discuss the implementation details and compare "\minOCA" with the method by Bruyère et al.~\cite{gaetan}, hereafter referred to as ""\text{\upshape\bps}@\bps"".

\subsubsection{Equivalence Query: }
Even though there is a polynomial time algorithm to check the "equivalence" of two "\drocas"~\cite{droca}, the polynomial is so large that it is not suitable for practical applications. 
In "\minOCA", we construct a "\droca" that is "counter-synchronised" with the "\droca" to be learnt. The "equivalence" is checked by a breadth-first search on the configuration graph up to the counter value and length obtained from \Cref{countersync}. 
The "minimal synchronous-equivalence query" either returns a word for which the constructed "\droca" and "\droca" to be learnt reach different counter values or returns a word that is accepted by one and rejected by the other. In our implementation, after each equivalence query, we increment the value of $d$ and make the table is "$d$-closed" and "$d$-consistent".

One major distinction between "\minOCA" and "\bps" is that the latter employs an approximate equivalence query, while the former uses an exact equivalence query. This implies that the "\droca" returned by "\bps" after learning may not be correct. On the other hand, the "\droca" returned by "\minOCA" is always correct. 

\subsubsection{Finding a Minimal-Separating DFA:}
We utilise the Python library \texttt{DFAMiner} by Dell'Erba et al.~\cite{dfaminer} that uses a SAT solver to find a minimal separating \dfa from a given set of positive and negative samples. 
Various other techniques for computing a minimal separating automaton can be found in \cite{Leucker,neidersat,dfa-identify}.

\subsubsection{Random Generation of DROCAs:}
\label{randomOCA}
We follow a procedure similar to that by  Bruyère et al.~\cite{gaetan} to randomly generate "\drocas" with a given number of states. 

Let $n\in\N$ be the number of states of the "\droca" to be generated. The procedure \texttt{GenerateDROCA} used to generate random "\drocas" is as follows. 
First, we initialise the set of states $Q=\{q_1,q_2,\ldots, q_n\}$.
For all $q\in Q$, we add $q$ to the set of final states $F$ with probability 0.5. If $Q=F$ or $F=\emptyset$ after this step, then we restart the procedure. Otherwise, for all $q\in Q$ and $a\in\Sigma$, we assign $\delta_0(q,a)=(p,c)$ (resp. $\delta_1 (q,a)=(p,c)$), with $p$ a random state in $Q$ and $c$ a random counter operation in $\{0,+1\}$ (resp. $\{0,+1,-1\}$). The constructed "\droca" { is} $\Autom=(Q, \Sigma, \{q_1\}, \delta_0,\delta_1, F)$.
If the number of reachable states of \Autom from the initial configuration is not $n$, then we discard $\Autom$ and restart the whole procedure. Otherwise, we output $\Autom$.
The exact procedure is given in \Cref{generateDROCA}. 

\label{algoGenerateRoca}
\begin{algorithm}[H]
\SetKwFunction{proc}{$\texttt{GenerateDROCA}$}{}{}
\SetKwProg{myproc}{Procedure}{}{end}
\SetKwInOut{Input}{Input}
\SetKwInOut{Output}{Output}
\SetKwFunction{connect}{Connect}
\SetKwData{ce}{ce}
\myproc{\proc{}}{
    \Input{An integer $n$}
    \Output{A random \droca with $n$ reachable states}
    \BlankLine
    \Repeat{$F\neq\emptyset$, $F\neq Q$ and \upshape\texttt{reachable}$(\Autom)=n$}{
    Initialise $Q=\{1,2,\ldots,n\}, F=\emptyset$.\\
    \textbf{foreach }{$p$ in $Q$}{
    	add $p$ to $F$ with probability $0.5$.
    }\\
    \textbf{foreach }{$p$ in $Q$ and $a\in\Sigma$}{
    assign $\delta_0(q,a)=(p,c)$ (resp. $\delta_1 (q,a)=(p,c)$), with $p$ a random state in $Q$ and $c$ a random counter operation in $\{0,+1\}$ (resp. $\{0,+1,-1\}$).
    }\\
	Initialise \droca $\Autom=(Q, \Sigma, \{1\}, \delta_0,\delta_1, F)$.
    }
   \textbf{return} $\Autom$.
}
\caption{Algorithm to generate a random \droca with $n$ states.}
\label{generateDROCA}
\end{algorithm}

To ensure that the generated "\droca" has $n$ reachable states, we use the procedure $\texttt{reachable}$ -- this procedure takes a "\droca" as input and returns the number of distinct states visited. It performs a breadth-first search on the configuration graph up to counter value $n^2$ starting from the initial configuration. If a state $q$ is reachable from the initial configuration, then there exists a word $w$ that takes us from that initial configuration to a configuration with state $q$, and the maximum counter value encountered during the run on $w$ is less than $n^2$.  Since the configuration graph up to counter value $n^2$ contains only $n^3$ configurations, this can be done in $\mathcal{O}(n^3)$ time. 
Two datasets, each with 5600 "\drocas", were generated, varying input alphabet sizes from 2 to 5 and states from 2 to 15, with 100 random "\drocas" for each combination.

\paragraph{\textbf{""\dsOne"": Dataset for comparing \upshape"\minOCA" \textit{and} "\bps".}}

The notion of acceptance in "\bps" is by final state and zero counter value, whereas "\minOCA" employs the notion of acceptance with the final state only. This is the widely accepted notion of acceptance for "\drocas" and is the one used by \cite{droca} while proving that the equivalence of "\drocas" is in $\CF{P}$. 
For comparing the two methods, we generated random "\drocas" where the notion of acceptance by final state and counter value zero and the notion of acceptance by final state are the same. For this, we use the procedure \texttt{GenerateDROCA} with an added condition.
This mandates that the final states can only be reached by transitions that read a symbol from counter value of zero and do not modify the counter value.
The results of experiments conducted using "\dsOne" on "\minOCA" and "\bps" are shown in Figures \ref{successFig}-\ref{columnFig}.

\paragraph{\textbf{""\dsTwo"": Dataset for evaluating the performance of \upshape"\minOCA" \textit{on general} DROCAs.}}
The performance of "\minOCA" was evaluated for general "\drocas", where the notion of acceptance is with final states only. Random \drocas were generated using the procedure \texttt{GenerateDROCA} for this purpose. 
The results of experiments conducted using "\dsTwo" on "\minOCA" are shown in Figures \ref{successPyFig}-\ref{columnPyFig}.

\subsubsection{Experimental Results:}

We implemented "\minOCA" in Python and used the Java implementation of "\bps". The computations were performed on an Apple M1 chip with 8GB of RAM, running macOS Sonoma Version 14.3. 

\subsubsection{Comparing \minOCA and \bps Using \dsOne.} 
\begin{figure}[!h]
\centering
\begin{minipage}[t]{\dimexpr.5\textwidth-1em}
\centering
\resizebox{\columnwidth}{!}{%
\includegraphics[trim=0cm 2cm 0cm 0cm, clip=true] {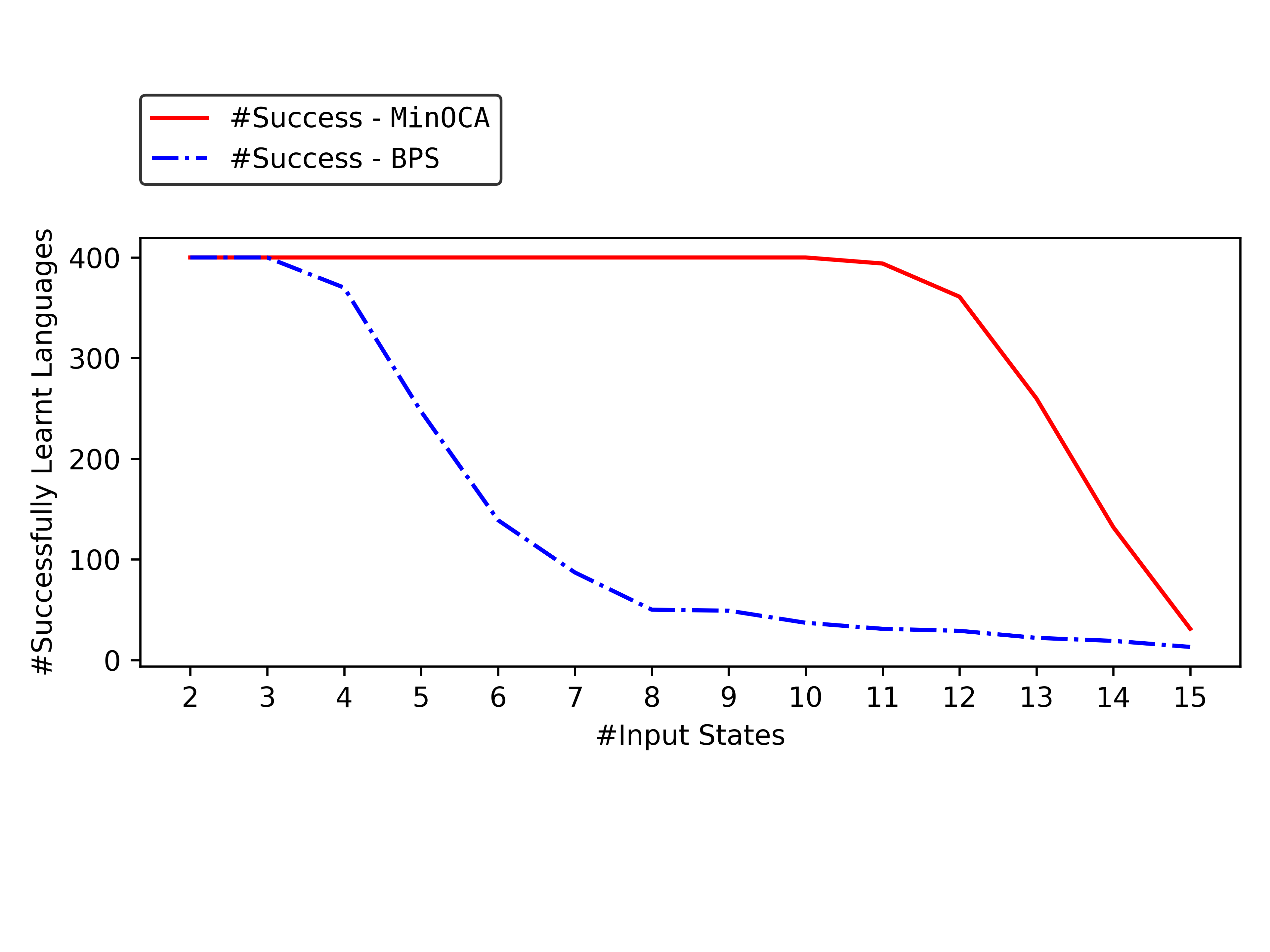}
}
\caption[Number of successfully learnt languages by \minOCA and \bps.]{Number of successfully learnt languages (out of 400) by "\minOCA" and "\bps" for "\dsOne".}
\label{successFig}
\end{minipage}%
\hfill
\begin{minipage}[t]{\dimexpr.5\textwidth-1em}
\centering
\resizebox{\columnwidth}{!}{%
\includegraphics[trim=0cm 2cm 0cm 0cm, clip=true]{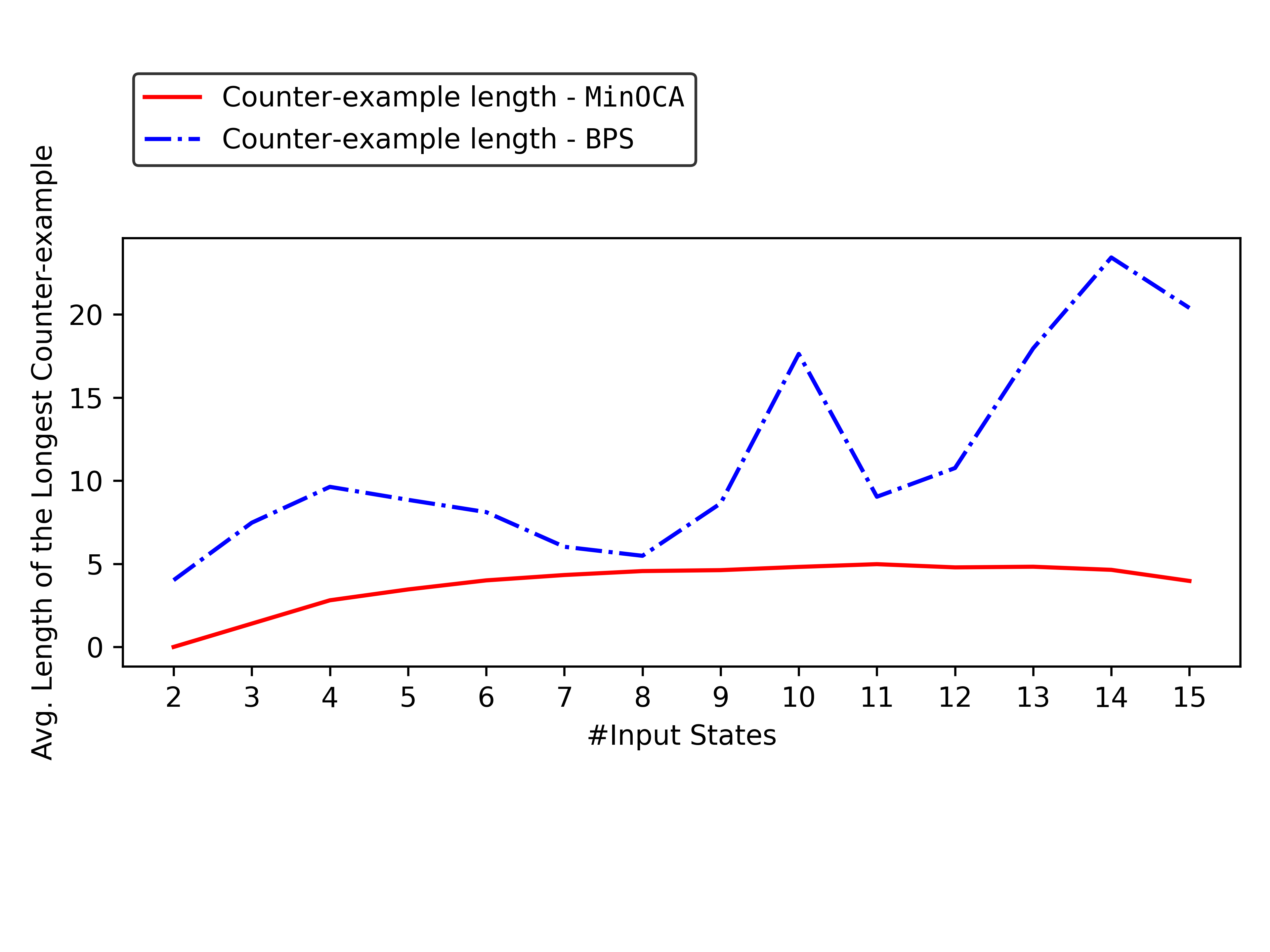}
}
\caption[The average length of the longest counter-example obtained by \minOCA and \bps for \dsOne.]{The average length of the longest counter-example obtained by "\minOCA" and "\bps" for "\dsOne".}
\label{ceFig}
\end{minipage}
\centering
\begin{minipage}{.47\textwidth}
\centering
\resizebox{\columnwidth}{!}{%
\includegraphics[trim=0cm 2cm 0cm 0cm, clip=true]{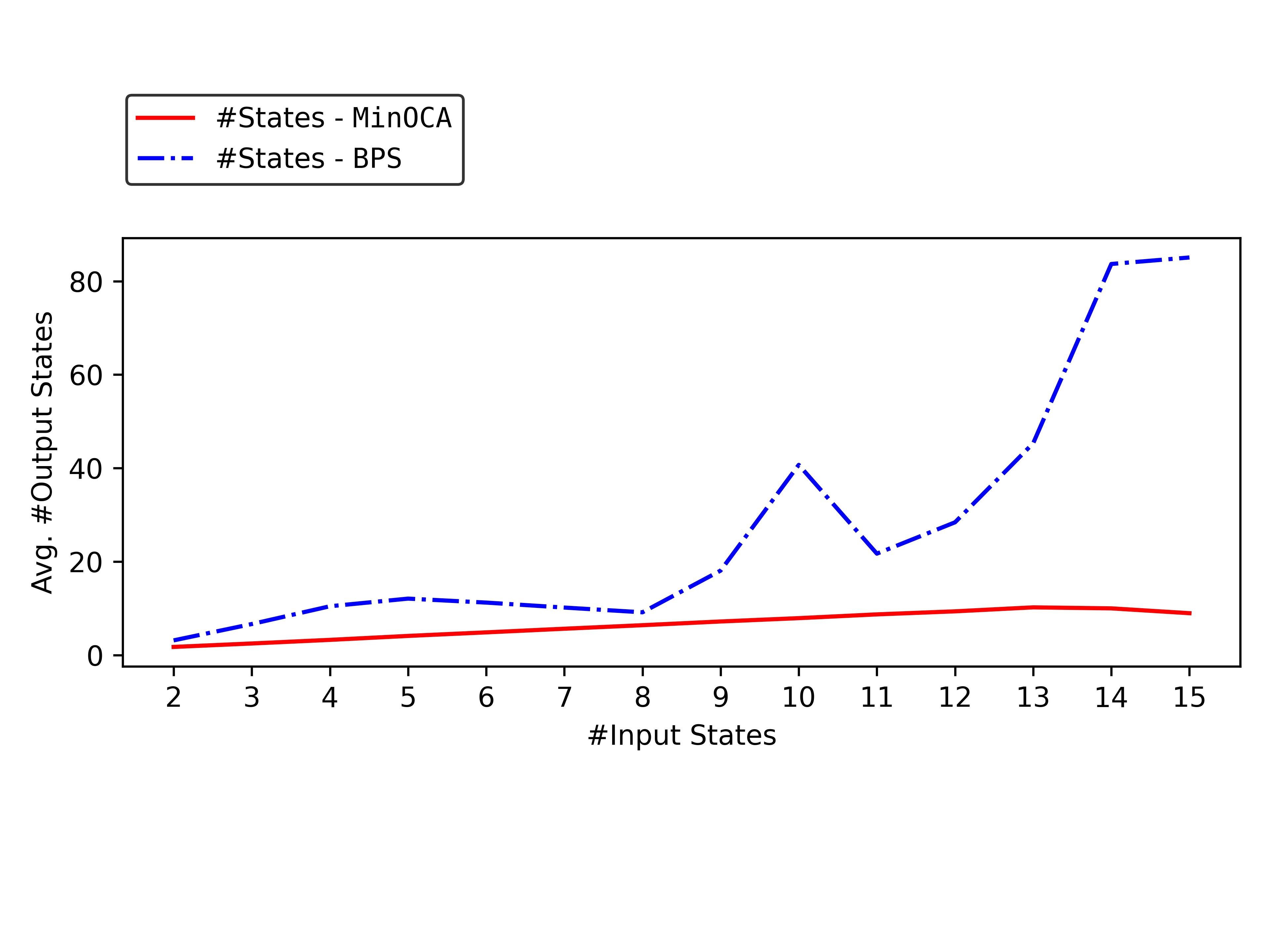}
}
\caption[Average number of states in the learnt \droca obtained by \minOCA and \bps for \dsOne.]{Average number of states in the learnt "\droca" obtained by "\minOCA" and "\bps" for "\dsOne".}
\label{StatesEquiv}
\end{minipage}%
\hfill
\begin{minipage}{.47\textwidth}
\centering
\hspace{5cm}
\resizebox{\columnwidth}{!}{%
\includegraphics[trim=0cm 2cm 0cm 0cm, clip=true]{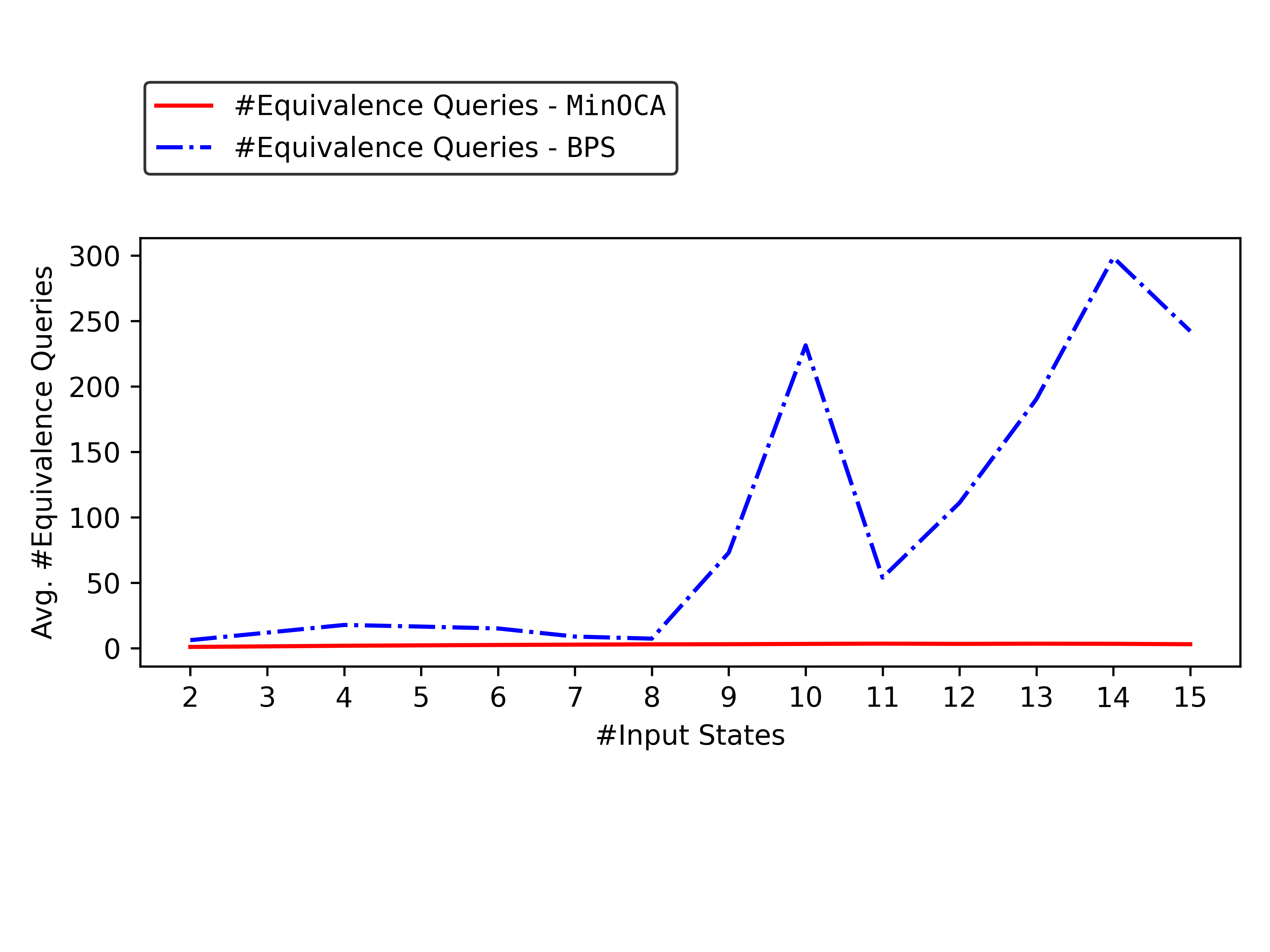}
}
\caption[Average number of equivalence queries used for learning by \minOCA and \bps for \dsOne.]{Average number of equivalence queries used for learning by "\minOCA" and "\bps" for "\dsOne".}
\label{avgEquiv}
\end{minipage}
\hfill
\begin{minipage}[t]{\dimexpr.5\textwidth-1em}
\centering
\resizebox{\columnwidth}{!}{%
\includegraphics[trim=0cm 2cm 0cm 0cm, clip=true] {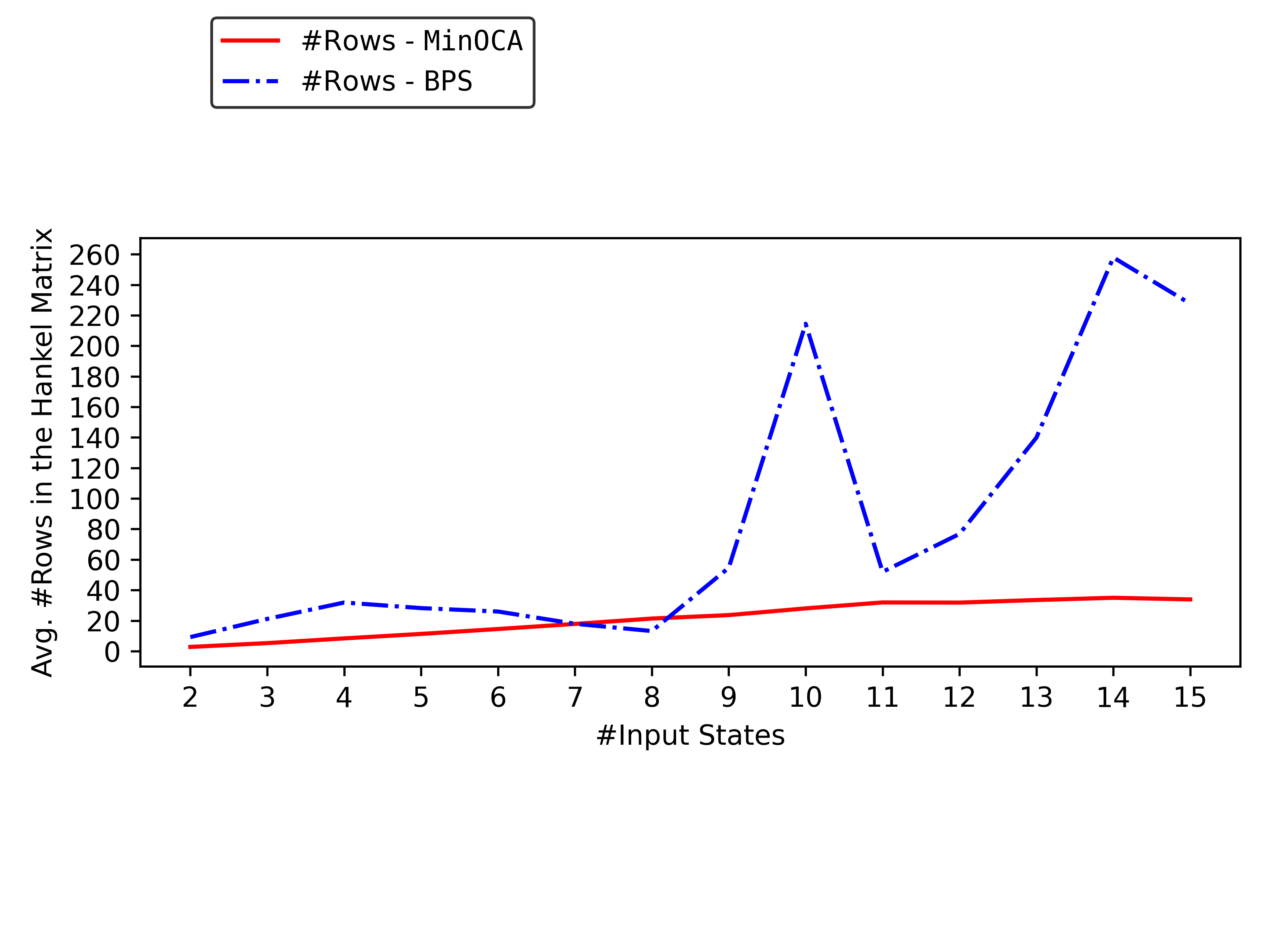}
}
\caption[Average number of rows in the observation table obtained by \minOCA and \bps for \dsOne.]{Average number of rows in the "observation table" obtained by "\minOCA" and "\bps" for "\dsOne".}
\label{rowsFig}
\end{minipage}%
\hfill
\begin{minipage}[t]{\dimexpr.5\textwidth-1em}
\centering
\resizebox{\columnwidth}{!}{%
\includegraphics[trim=0cm 2cm 0cm 0cm, clip=true]{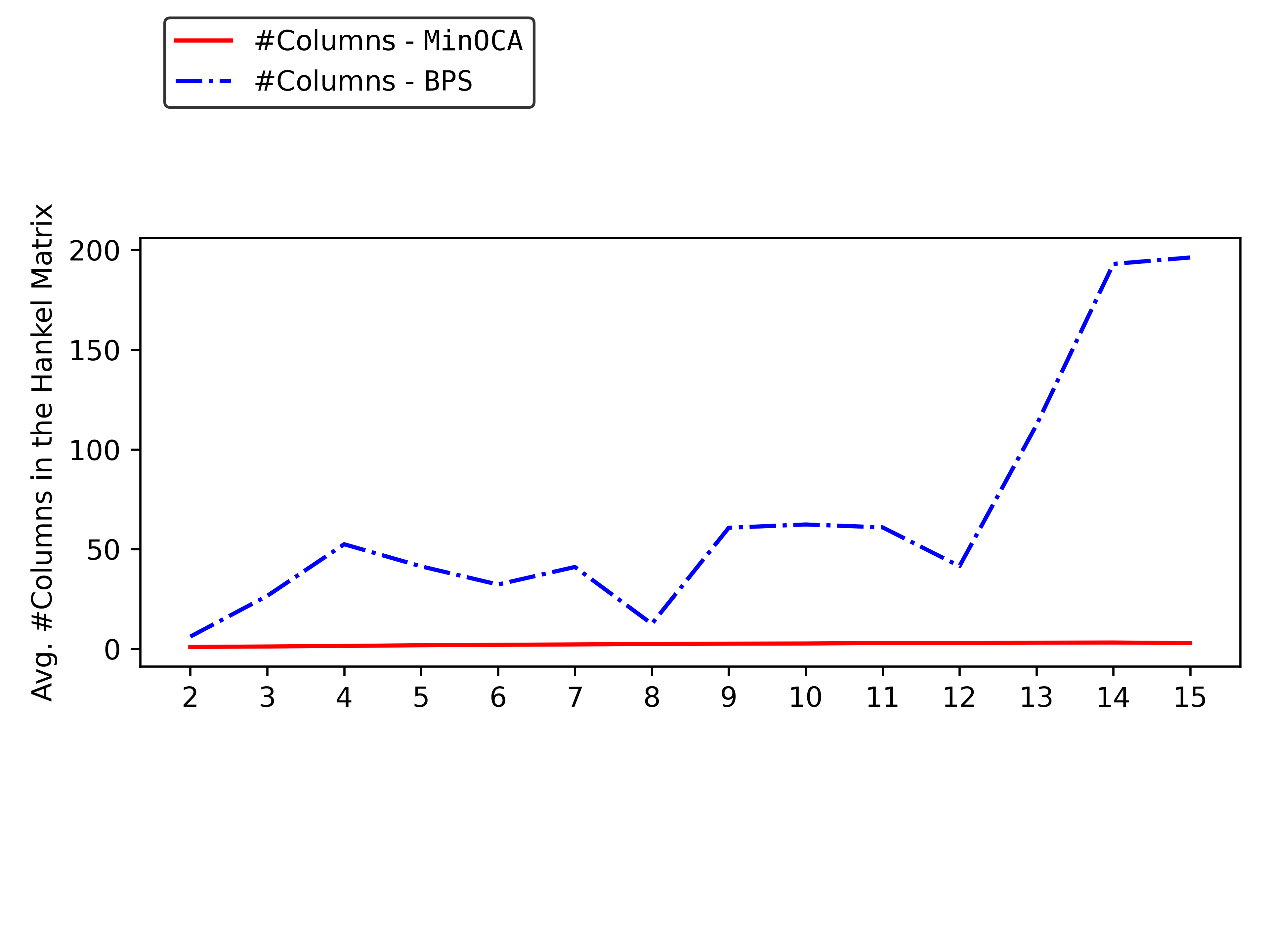}
}
\caption[Average columns in the observation table for \minOCA and \bps on \dsOne.]{Average columns in the "observation table" for "\minOCA" and "\bps" on "\dsOne".}
\label{columnFig}
\end{minipage}
\end{figure}

A timeout of $5$ minutes was allotted for both "\bps" and "\minOCA" for learning each "\droca" in "\dsOne". If the algorithms cannot successfully learn a language within the timeout, we discard that sample and process the next one. The number of languages successfully learned by "\minOCA" and "\bps" for different input sizes is depicted in \Cref{successFig}, and the exact values are provided in \Cref{compTable}. 
 \Cref{ceFig} shows the average length of the longest counter-example obtained during learning. In all cases, "\minOCA" provides a smaller counter-example on average.  \Cref{StatesEquiv} shows the average number of states in the learnt "\droca". The number of states of the automaton learnt by "\minOCA" is always less than or equal to the input size. 
 \Cref{avgEquiv} shows the average number of equivalence queries used for successfully learning the input "\drocas". The number of equivalence queries is smaller for "\minOCA" in all cases. The data used to plot Figures \ref{successFig}- \ref{columnFig} is given in \Cref{fullData1}.

\begin{table}[h]
\centering\scalebox{1.1}{ 
\begin{tabular}{|l|c|l|l|l|l|l|l|l|l|l|l|l|l|l|l|}
\hline
                        & \multicolumn{1}{l|}{\backslashbox{$|\Sigma|$}{$|Q|$}} & \multicolumn{1}{c|}{2} & \multicolumn{1}{c|}{3} & \multicolumn{1}{c|}{4} & \multicolumn{1}{c|}{5} & \multicolumn{1}{c|}{6} & \multicolumn{1}{c|}{7} & \multicolumn{1}{c|}{8} & \multicolumn{1}{c|}{9} & \multicolumn{1}{c|}{10} & \multicolumn{1}{c|}{11} & \multicolumn{1}{c|}{12} & \multicolumn{1}{c|}{13} & \multicolumn{1}{c|}{14} & \multicolumn{1}{c|}{15}\\ \Xhline{1pt}
\multirow{4}{*}{\rotatebox{90}{\minOCA}} & 2                     & 100                    & 100                    & 100                    & 100                    & 100                    & 100                     & 100                     & 100                     & 100           & 95 & 85 & 54 & 27 & 7           \\ \cline{2-16} 
                        & 3                     & 100                    & 100                    & 100                    & 100                    & 100                     & 100                     & 100                     & 100                     & 100                   & 99  & 91  & 61  & 22 & 1      \\ \cline{2-16} 
                        & 4                     & 100                    & 100                    & 100                    & 100                     & 100                     & 100                    & 100                     & 100                     & 100                    & 100 & 91 & 79 &38  & 9      \\ \cline{2-16} 
                        & 5                     & 100                    & 100                    & 100                    & 100                    & 100                     & 100                     & 100                     & 100                      & 100                  & 100  & 94 & 66 & 45 & 14        \\ \Xhline{1pt}
\multirow{4}{*}{\rotatebox{90}{"\bps"}}    & 2                     & 100                    & 100                    & 98                     & 80                     & 61                     & 48                     & 39                     & 32                     & 32                  & 24 &  23 & 18  & 15  &   9     \\ \cline{2-16} 
                        & 3                     & 100                    & 100                    & 95                     & 67                     & 31                     & 17                     & 5                      & 12                     & 2                       & 5 &5  & 2 & 2 &  4  \\ \cline{2-16} 
                        & 4                     & 100                    & 100                    & 90                     & 48                     & 25                     & 13                     & 4                      & 3                      & 3                       & 2 & 0 & 1 & 2 &  0  \\ \cline{2-16} 
                        & 5                     & 100                    & 100                    & 87                     & 52                     & 22                     & 9                      & 2                      & 2                      & 0                       &  0& 1 & 1 & 0 &  0  \\ \hline
\end{tabular}
}
\caption[Number of successfully learnt samples by \minOCA and \bps.]{Number of successfully learnt samples (out of 100) by "\minOCA" and "\bps" for "\dsOne" with the number of states ranging from $2$ to $15$ and the size of the alphabet ranging from $2$ to $5$.}
\label{compTable} 
\end{table}

\begin{table}[h] 
\centering
\begin{tabular}{|l|l|l|l|l|l|l|l|l|l|l|l|l|}
\hline
\textbf{\#States} & \textbf{Success1} & \textbf{Success2} & \textbf{States1} & \textbf{States2} & \textbf{LCE1} & \textbf{LCE2} & \textbf{EqQ1} & \textbf{EqQ2}& \textbf{Row1} & \textbf{Row2} & \textbf{Col1} & \textbf{Col2} \\ \hline
2 & 400 & 400 & 1.75 & 3.15 & 0.00 & 4.02 & 1.00 & 6.19 & 2.78 & 9.22 & 1.00 & 6.18 \\
3 & 400 & 400 & 2.49 & 6.65 & 1.40 & 7.47 & 1.43 & 11.93 & 5.28 & 21.13 & 1.21 & 26.59 \\
4 & 400 & 370 & 3.26 & 10.45 & 2.81 & 9.63 & 1.91 & 17.79 & 8.38 & 31.86 & 1.51 & 52.49 \\
5 & 400 & 247 & 4.11 & 12.09 & 3.46 & 8.85 & 2.21 & 16.55 & 11.30 & 28.20 & 1.83 & 41.39 \\
6 & 400 & 139 & 4.87 & 11.24 & 4.00 & 8.11 & 2.50 & 15.11 & 14.52 & 25.96 & 2.06 & 32.29 \\
7 & 400 & 87 & 5.64 & 10.17 & 4.33 & 6.02 & 2.74 & 8.93 & 17.84 & 17.95 & 2.25 & 41.08 \\
8 & 400 & 50 & 6.40 & 9.16 & 4.56 & 5.48 & 2.98 & 7.32 & 21.36 & 13.18 & 2.46 & 12.56 \\
9 & 400 & 49 & 7.20 & 18.08 & 4.62 & 8.63 & 3.09 & 73.01 & 23.62 & 54.61 & 2.64 & 60.73 \\
10 & 400 & 37 & 7.91 & 40.70 & 4.81 & 17.62 & 3.31 & 231.46 & 28.03 & 214.43 & 2.72 & 62.38 \\
11 & 394 & 31 & 8.71 & 21.71 & 4.98 & 9.03 & 3.51 & 53.97 & 31.92 & 51.94 & 2.95 & 60.94 \\
12 & 361 & 29 & 9.37 & 28.41 & 4.78 & 10.76 & 3.32 & 111.10 & 31.83 & 76.79 & 2.91 & 41.48 \\
13 & 260 & 22 & 10.22 & 45.32 & 4.82 & 17.95 & 3.47 & 190.36 & 33.50 & 139.86 & 3.12 & 112.18 \\
14 & 132 & 19 & 10.00 & 83.68 & 4.64 & 23.42 & 3.38 & 298.58 & 34.98 & 257.89 & 3.20 & 193.00 \\
15 & 31 & 13 & 8.97 & 85.08 & 3.97 & 20.38 & 3.06 & 242.31 & 33.90 & 227.62 & 2.94 & 196.23 \\ \hline
\end{tabular}
\captionof{table}[Data used to plot Figures \ref{successFig}- \ref{columnFig}.]{{The table shows the data used to plot Figures \ref{successFig}- \ref{columnFig}. The column \emph{\#States} denotes the number of states in the input "\droca". The columns $Sucess1,\ States1,\ LCE1,\ EqQ1,\ Row1,\ Col1$ (resp. $Sucess2,\ States2,\ LCE2,\ EqQ2,\ Row2,\ Col2$) respectively denote the number of successfully learnt languages, the average number of states in the learnt "\droca", the average length of the longest counter-example, the average number of equivalence queries used, the average number of rows in the final "observation table" and average number of columns in the final "observation table" for "\minOCA" (resp. "\bps") for "\dsOne".}}
\label{fullData1}
\end{table}

\subsubsection{Evaluating the Performance of \minOCA on \dsTwo.}
\label{generalDROCA}
We used "\dsTwo" to evaluate the performance of "\minOCA" on "\drocas" that accept with final state.
A timeout of $5$ minutes was allocated for learning each language. \Cref{successPyFig} shows the number of successfully learnt languages by "\minOCA". One can see that as the number of states exceed $10$, there is a decrease in the number of samples learnt in $5$ minutes. 

The number of languages successfully learned by "\minOCA" for different input sizes is depicted in \Cref{successPyFig}, and the exact values are provided in \Cref{compTable2}. 
 \Cref{cePyFig} shows the average length of the longest counter-example obtained during learning.  \Cref{StatesPyEquiv} shows the average number of states in the learnt "\droca".
 \Cref{avgPyEquiv} shows the average and maximum number of equivalence queries used for successfully learning the input "\drocas".  The data used to plot Figures \ref{successPyFig}- \ref{columnPyFig} is given in \Cref{fullData2}.

\begin{figure}
\centering
\begin{minipage}[t]{\dimexpr.5\textwidth-1em}
\centering
\resizebox{\columnwidth}{!}{%
\includegraphics[trim=0cm 2cm 0cm 0cm, clip=true] {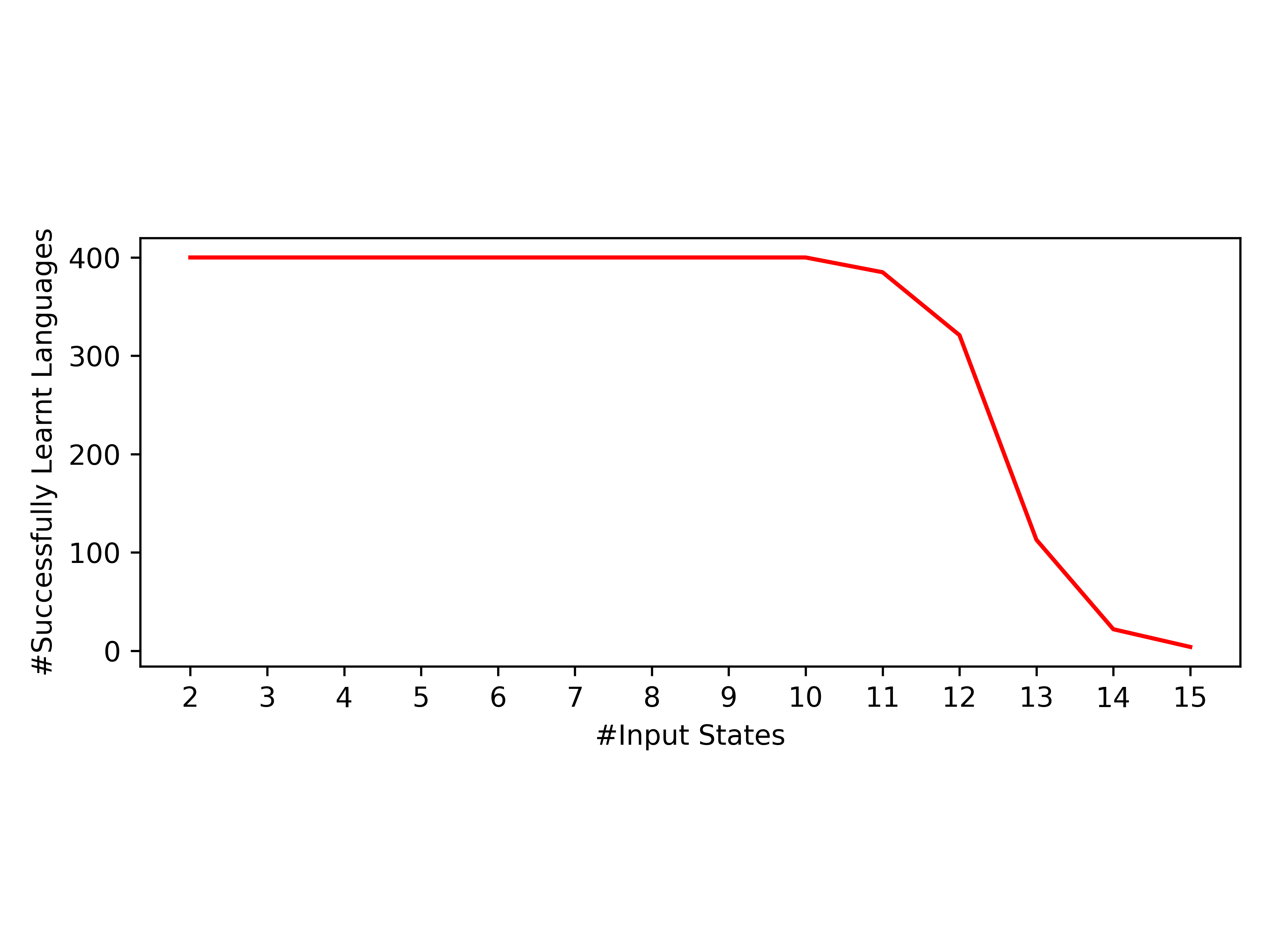}
}
\caption[Number of successfully learnt languages by \minOCA.]{Number of successfully learnt languages (out of 400) by "\minOCA" for "\dsTwo".}
\label{successPyFig}
\end{minipage}%
\hfill
\begin{minipage}[t]{\dimexpr.5\textwidth-1em}
\centering
\resizebox{\columnwidth}{!}{%
\includegraphics[trim=0cm 2cm 0cm 0cm, clip=true]{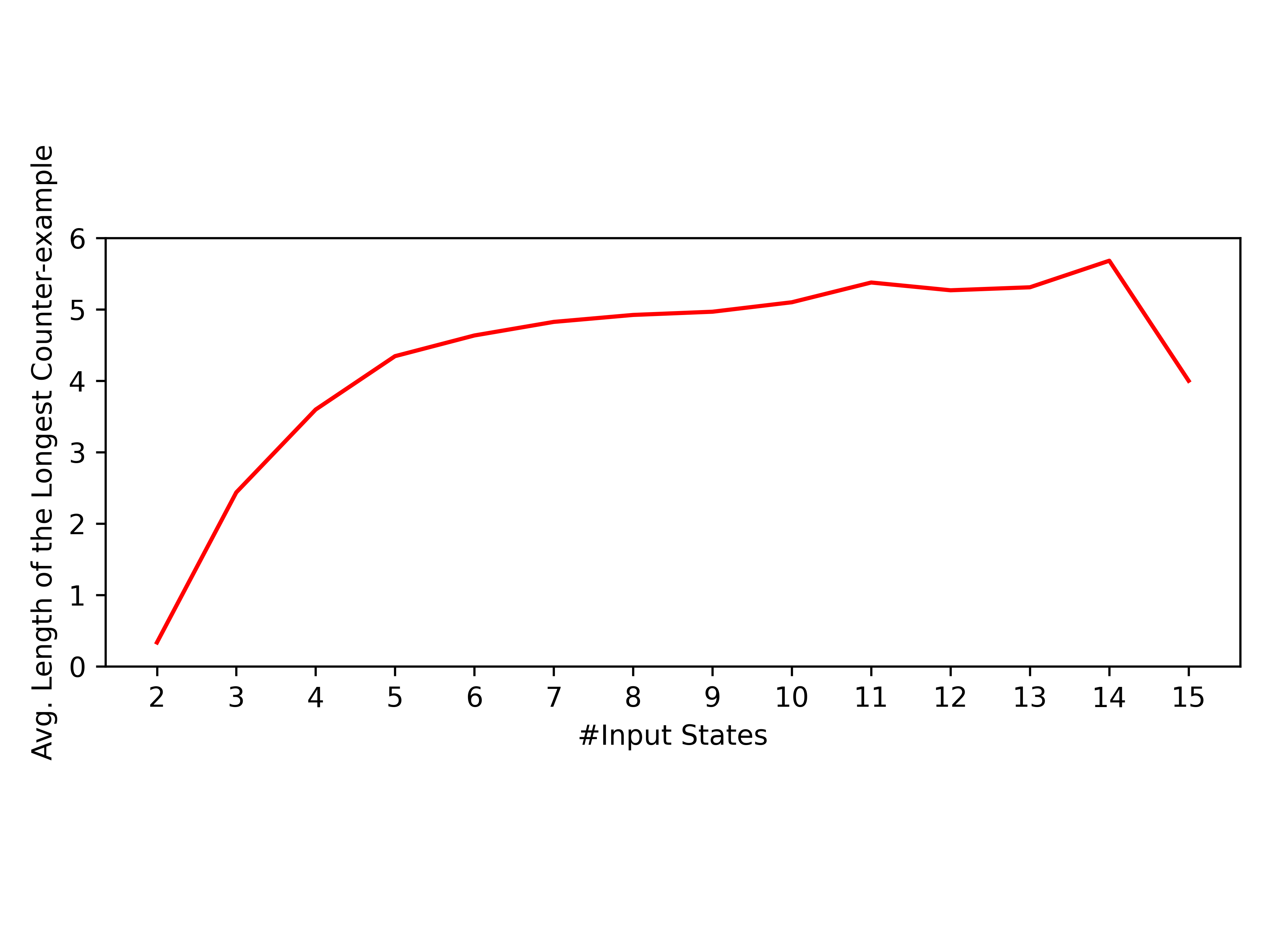}
}
\caption[The average length of the longest counter-example obtained by \minOCA for \dsTwo.]{The average length of the longest counter-example obtained by "\minOCA" for "\dsTwo".}
\label{cePyFig}
\end{minipage}
\centering
\begin{minipage}{.47\textwidth}
\centering
\resizebox{\columnwidth}{!}{%
\includegraphics[trim=0cm 2cm 0cm 0cm, clip=true]{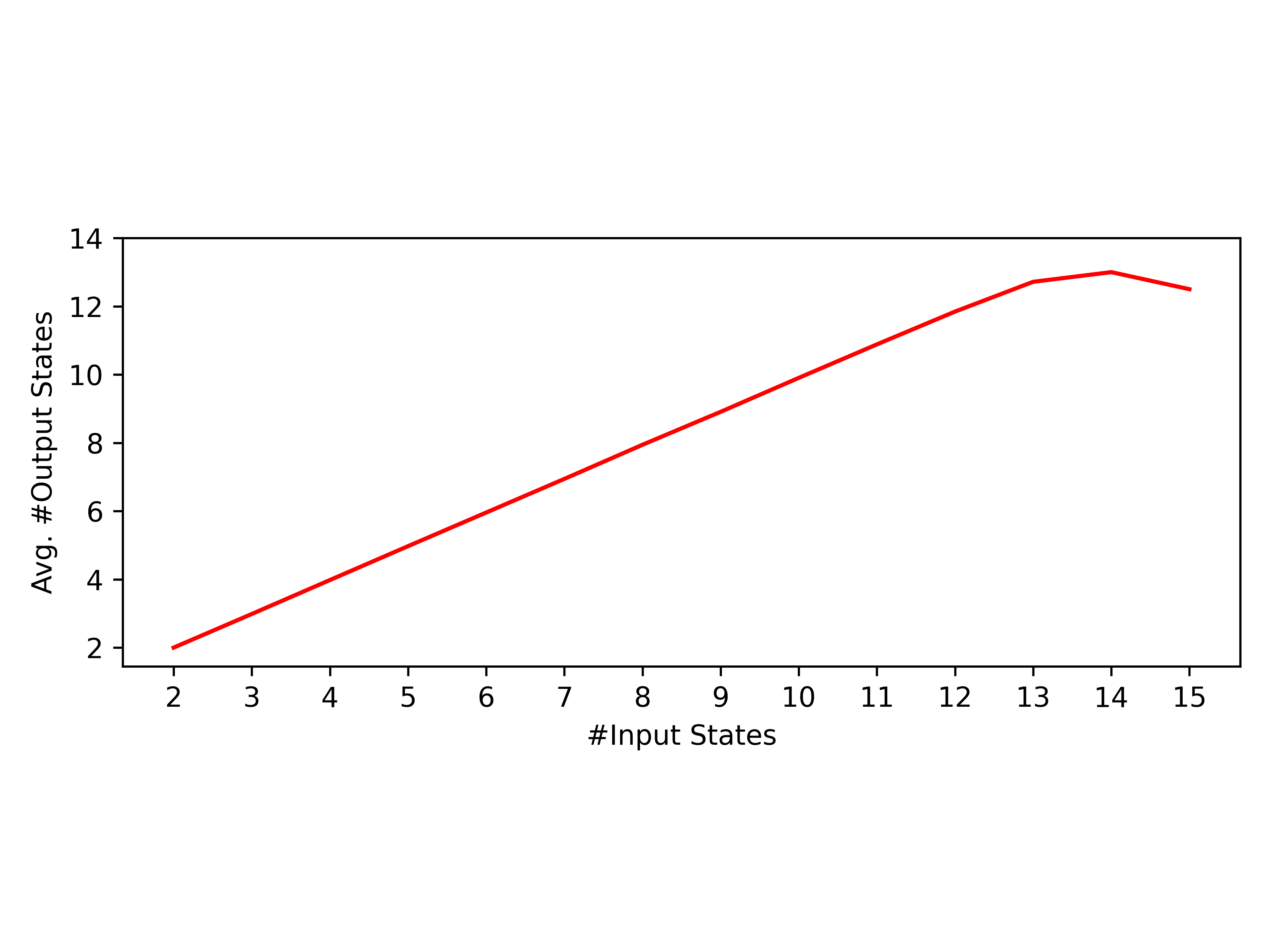}
}
\caption[Average number of states in the learnt \droca obtained by \minOCA for \dsTwo.]{Average number of states in the learnt "\droca" obtained by "\minOCA" for "\dsTwo".}
\label{StatesPyEquiv}
\end{minipage}%
\hfill
\begin{minipage}{.47\textwidth}
\centering
\hspace{5cm}
\resizebox{\columnwidth}{!}{%
\includegraphics[trim=0cm 2cm 0cm 0cm, clip=true]{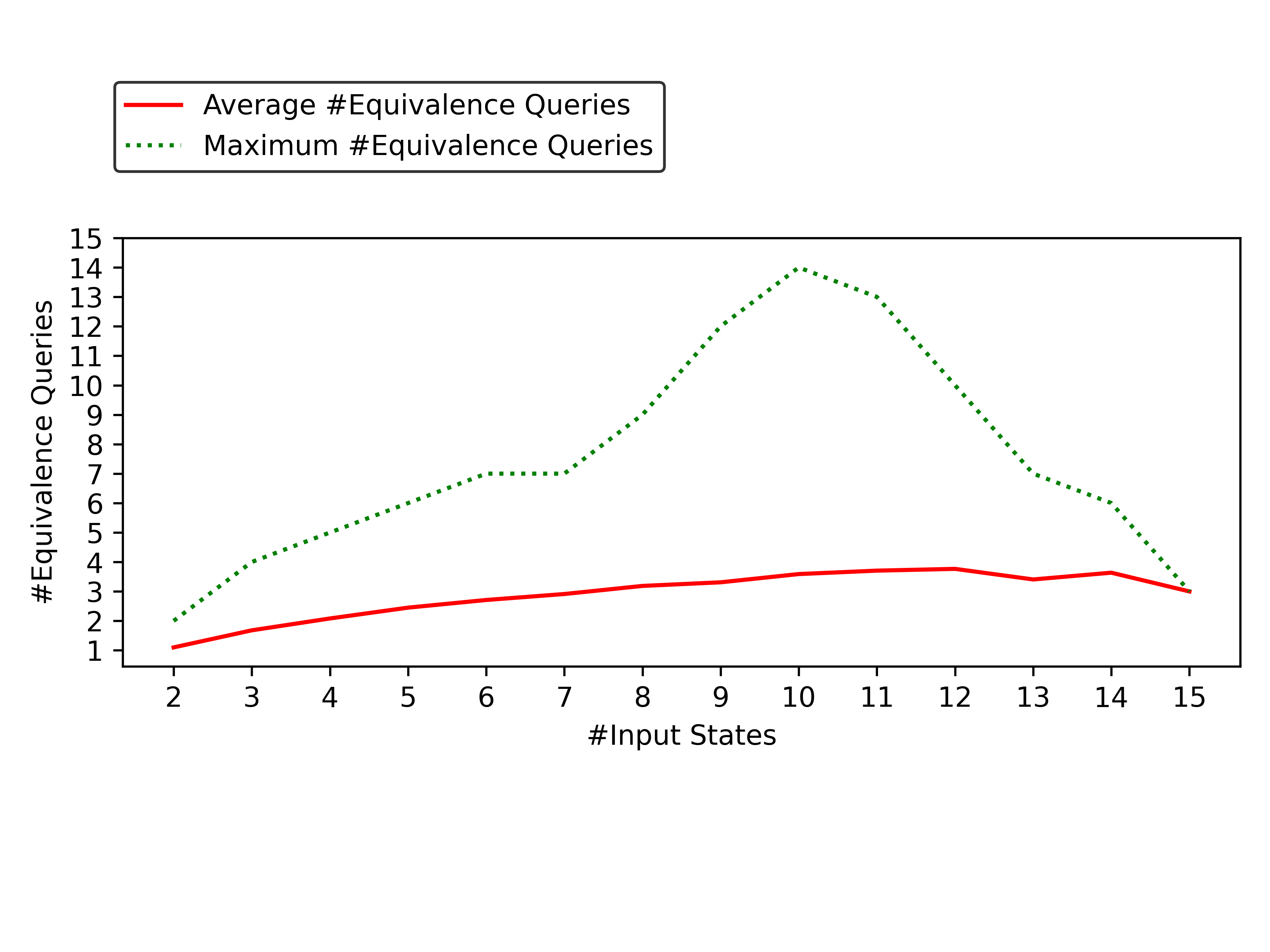}
}
\caption[Average and maximum number of equivalence queries used for learning by \minOCA for \dsTwo.]{Average and maximum number of equivalence queries used for learning by "\minOCA" for "\dsTwo".}
\label{avgPyEquiv}
\end{minipage}
\hfill
\begin{minipage}[t]{\dimexpr.5\textwidth-1em}
\centering
\resizebox{\columnwidth}{!}{%
\includegraphics[trim=0cm 2cm 0cm 0cm, clip=true] {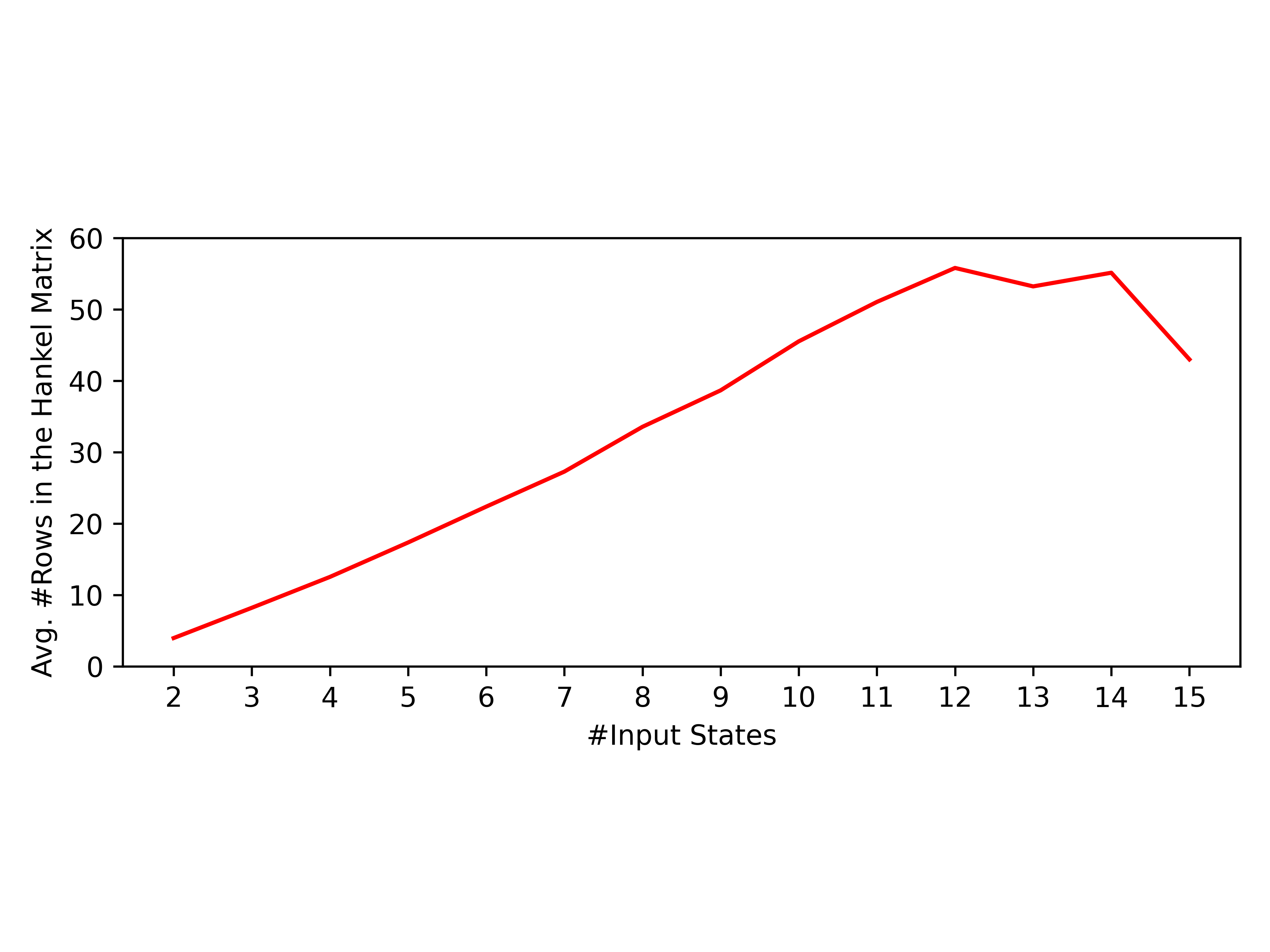}
}
\caption[The average number of rows in the observation table obtained by \minOCA for \dsTwo.]{The average number of rows in the "observation table" obtained by "\minOCA" for "\dsTwo".}
\label{rowsPyFig}
\end{minipage}%
\hfill
\begin{minipage}[t]{\dimexpr.5\textwidth-1em}
\centering
\resizebox{\columnwidth}{!}{%
\includegraphics[trim=0cm 2cm 0cm 0cm, clip=true]{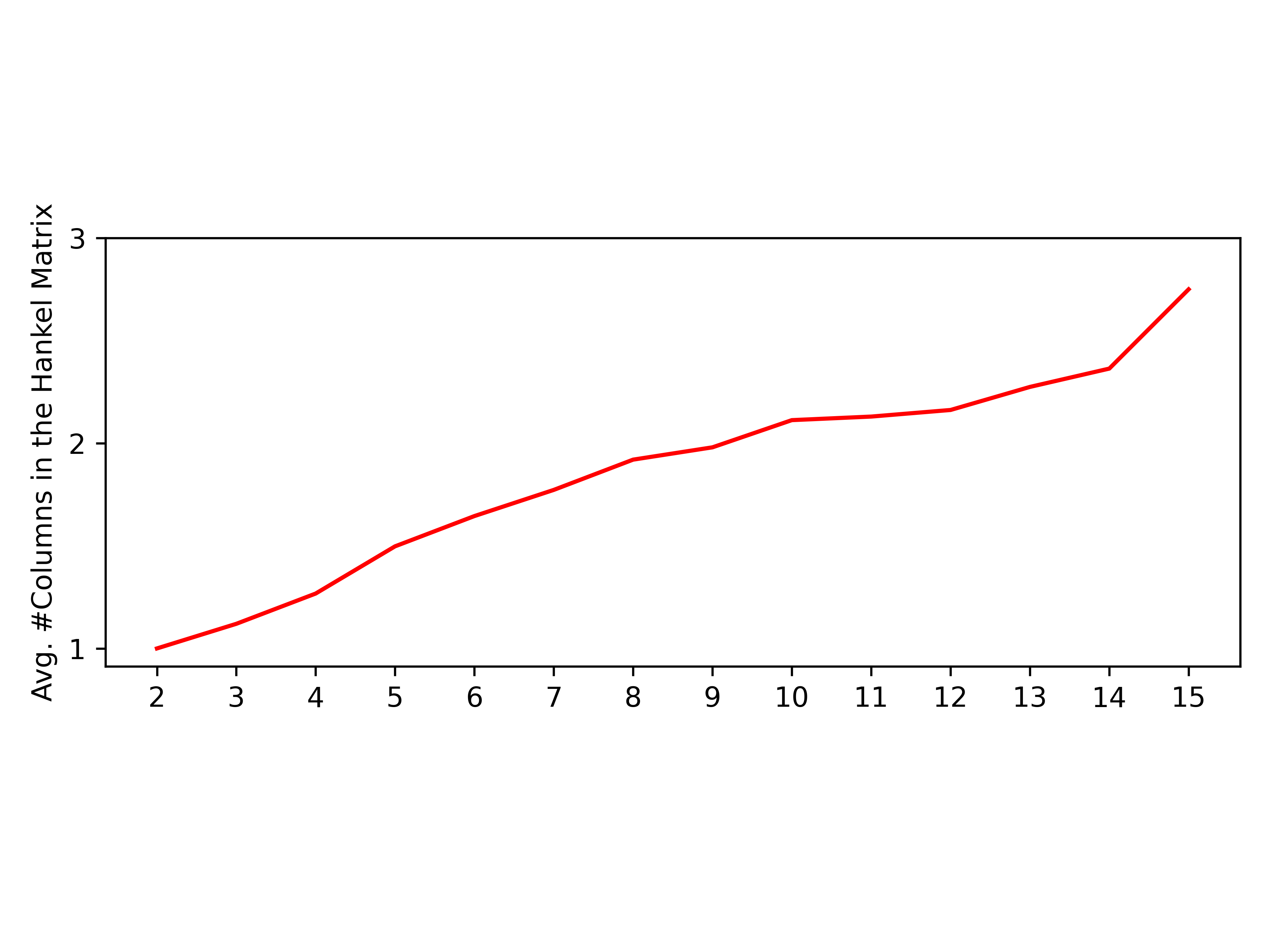}
}
\caption[The average number of columns in the observation table obtained by \minOCA for \dsTwo.]{The average number of columns in the "observation table" obtained by "\minOCA" for "\dsTwo".}
\label{columnPyFig}
\end{minipage}
\end{figure}

\minOCA spends most of the time in finding a "minimal separating \dfa" (see \Cref{SatTime}) using the SAT solver. 
The scalability of our algorithm is hence dependent on the scalability of finding a "minimal separating \dfa".

\begin{table}[h]
\centering\scalebox{1.2}{ 
\begin{tabular}{|c|l|l|l|l|l|l|l|l|l|l|l|l|l|l|}
\hline
\backslashbox{$|\Sigma|$}{$|Q|$} & \multicolumn{1}{c|}{2} & \multicolumn{1}{c|}{3} & \multicolumn{1}{c|}{4} & \multicolumn{1}{c|}{5} & \multicolumn{1}{c|}{6} & \multicolumn{1}{c|}{7} & \multicolumn{1}{c|}{8} & \multicolumn{1}{c|}{9} & \multicolumn{1}{c|}{10} & \multicolumn{1}{c|}{11} & \multicolumn{1}{c|}{12} & \multicolumn{1}{c|}{13} & \multicolumn{1}{c|}{14} & \multicolumn{1}{c|}{15}\\ \Xhline{1pt}
2 & 100 & 100 & 100 & 100 & 100 & 100 & 100 & 100 & 100 & 94 & 77 & 34 & 10 & 3 \\ \hline
3 & 100 & 100 & 100 & 100 & 100 & 100 & 100 & 100 & 100 & 91 & 74 & 20 & 3 & 1 \\ \hline
4 & 100 & 100 & 100 & 100 & 100 & 100 & 100 & 100 & 100 & 100 & 83 & 33 & 4 & 0 \\ \hline
5 & 100 & 100 & 100 & 100 & 100 & 100 & 100 & 100 & 100 & 100 & 87 & 26 & 5 & 0 \\ \hline
\end{tabular}
}
\caption[Number of successfully learnt samples by \minOCA.]{Number of successfully learnt samples (out of 100) by "\minOCA" for "\dsTwo" with the number of states ranging from $2$ to $15$ and the size of the alphabet ranging from $2$ to $5$.}
\label{compTable2} 
\end{table}

\begin{table}[ht]
\centering\scalebox{1.1}{
\begin{tabular}{|c|c|c|c|c|c|c|c|c|c|}
\hline
\textbf{\small \#States} & \textbf{\small Success} & \textbf{\small States} & \textbf{\small LongestCE} & \textbf{\small EqQ} & \textbf{\small MaxEqQ} & \textbf{\small Time} & \textbf{\small SAT} & \textbf{\small Row} & \textbf{\small Col} \\
\hline
2   & 400  & 2.00  & 0.34  & 1.10  & 2.00  & 1.26  & 1.26  & 3.97  & 1.00  \\
3   & 400  & 2.99  & 2.44  & 1.68  & 4.00  & 1.93  & 1.93  & 8.21  & 1.12  \\
4   & 400  & 3.98  & 3.60  & 2.08  & 5.00  & 2.44  & 2.44  & 12.53 & 1.27  \\
5   & 400  & 4.97  & 4.35  & 2.45  & 6.00  & 2.92  & 2.91  & 17.35 & 1.50  \\
6   & 400  & 5.96  & 4.64  & 2.71  & 7.00  & 3.33  & 3.31  & 22.38 & 1.65  \\
7   & 400  & 6.94  & 4.83  & 2.91  & 7.00  & 3.91  & 3.86  & 27.28 & 1.77  \\
8   & 400  & 7.94  & 4.92  & 3.19  & 9.00  & 5.11  & 5.01  & 33.56 & 1.92  \\
9   & 400  & 8.91  & 4.97  & 3.31  & 12.00 & 7.62  & 7.42  & 38.66 & 1.98  \\
10  & 400  & 9.90  & 5.10  & 3.59  & 14.00 & 16.60 & 16.22 & 45.53 & 2.11  \\
11  & 385  & 10.88 & 5.38  & 3.71  & 13.00 & 45.34 & 44.69 & 51.05 & 2.13  \\
12  & 321  & 11.84 & 5.27  & 3.77  & 10.00 & 112.92 & 111.84 & 55.81 & 2.16  \\
13  & 113  & 12.72 & 5.31  & 3.41  & 7.00  & 169.28 & 167.69 & 53.22 & 2.27  \\
14  & 22   & 13.00 & 5.68  & 3.64  & 6.00  & 184.53 & 182.56 & 55.14 & 2.36  \\
15  & 4    & 12.50 & 4.00  & 3.00  & 3.00  & 135.98 & 135.00 & 43.00 & 2.75  \\
\hline
\end{tabular}
}
\captionsetup{width=.65\textheight } 
\caption[Data used to plot Figures \ref{successPyFig}- \ref{columnPyFig}.]{The table shows the data used to plot Figures \ref{successPyFig}- \ref{columnPyFig}. The column \emph{\#States} denotes the number of states in the input "\droca". The columns $Sucess,\ States,\ LongestCE,\ EqQ,\ MaxEqQ,\ Time,\ SAT,\ Row,\ Col$ respectively denote the number of successfully learnt languages, the average number of states in the learnt "\droca", the average length of the longest counter-example, the average number of "equivalence queries" used, the maximum number of "equivalence queries" used, the average time taken for learning, the average time taken by the SAT solver in finding a "minimal separating \dfa", the average number of rows in the final "observation table" and average number of columns in the final "observation table" for "\dsTwo".}
\label{fullData2}
\end{table}

\begin{figure}[!h]
\centering
\scalebox{.6}{%
\includegraphics{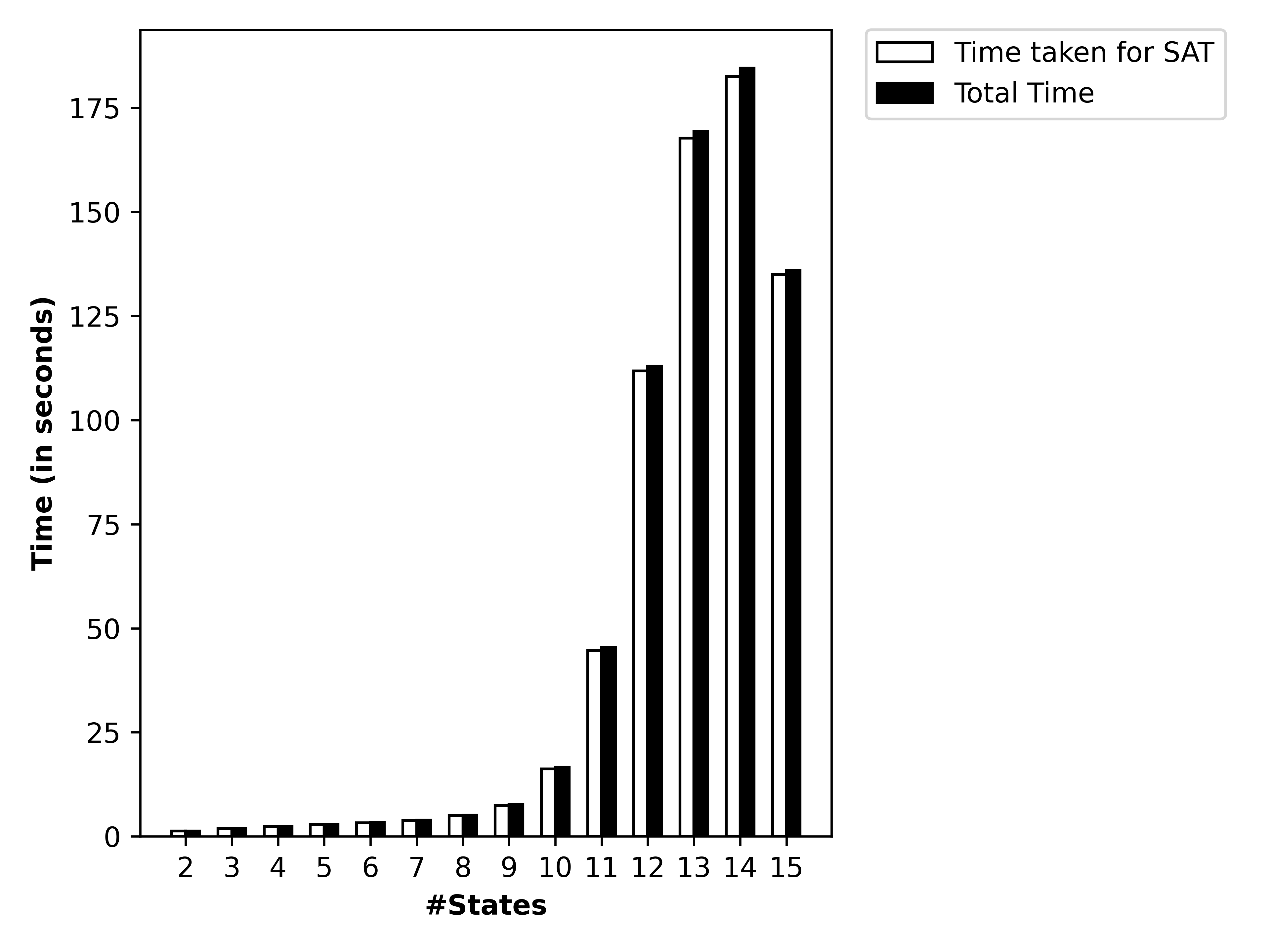}
}
\caption[Total time compared to the time taken by SAT solver.]{\centering Total time taken compared to the time taken by the SAT solver to find a "minimal separating \dfa" when tested on "\dsTwo".}
\label{SatTime}
\end{figure}

\label{compareAut}
\begin{figure}[htbp]
  { 
   \begin{subfigure}[b]{0.24\textwidth}   
   \Large
  \centering\scalebox{.45}{
\begin{tikzpicture}[shorten >=1pt,node distance=3cm,on grid,auto]
\tikzset{every path/.style={line width=.6mm}}\
\node[state] at (.2,12) (q_0) {$q_0$};
\node[state] at (2,9.5) (q_1) {$q_1$};
\node[state,accepting] at (4.5,5)(q_3){$q_3$};
\node[state] at (2,2)(q_2){$q_2$};
\path[->]
(q_0) edge [loop right] node [yshift=.2cm]{$a_{=0}/{\small+1}$} ()
edge [bend right, below] node[xshift=-.9cm,yshift=.3cm]  {$b_{=0}/{\small+1}$} (q_2)
edge [above] node[xshift=.7cm,yshift=-.2cm] {$b_{>0}/{\small-1}$} (q_1)
(q_1) edge [loop right] node {$b_{>0}/{\small-1}$} ()
edge [right] node[xshift=-1.8cm, yshift=0cm] {$a_{>0}/{\small+1}\ b_{=0}/{\small+1}$} (q_2)
edge [above] node[xshift=.5cm, yshift=0cm] {$a_{=0}/{\small0}$} (q_3)
(q_2) edge [loop left] node[xshift=-.1cm, yshift=.2cm]{$a_{>0}/{\small+1}$} ()
(q_3) edge [right] node[yshift=.3cm,xshift=.1cm]{$a_{=0}/{\small+1}$} (q_2);
\draw (2.2,11.8) node[xshift=.3cm] {$a_{>0}/{\small+1}$};
\draw (-.5,1.7) node {$b_{>0}/{\small+1}$};
\draw (4,3.3) node{$b_{=0}/{\small+}1$};
\draw[<-] (-.33,12)-- (-.85,12);
\end{tikzpicture}}
\caption{\centering Input "\droca".}
\label{anbna}
\end{subfigure}
}
\hfill
{
   \begin{subfigure}[b]{0.35\textwidth}   
   \Large
  \centering\scalebox{.45}{
\begin{tikzpicture}[shorten >=1pt,node distance=3cm,on grid,auto]
\tikzset{every path/.style={line width=.6mm}}\
\node[state] at (0,12) (q_1) {$q_1$};
\node[state] at (1.8,8) (q_3) {$q_3$};
\node[state,accepting] at (4,5)(q_2){$q_2$};
\node[state] at (3,16)(q_0){$q_0$};
\path[->]
(q_0) edge [loop right] node [yshift=.3cm]{$a_{>0}/{\small+1}$} ()
edge [bend right, below] node[xshift=-1.1cm,yshift=.3cm]  {$a_{=0}/{\small+1}$} (q_1)
(q_1) edge [loop left] node[xshift=1cm, yshift=.5cm] {$a_{>0}/{\small+1}$} ()
edge [right] node[xshift=-1.8cm, yshift=-.2cm] {$b_{>0}/{\small-1}$} (q_3)
(q_2) edge [bend right, left] node[xshift=1.6cm,yshift=-3.5cm]{$a_{=0}/{\small+1}\ \ b_{=0}/{\small+}1$} (q_0)
(q_3) edge [loop left] node {$b_{>0}/{\small-1}$} ()
edge [right] node[xshift=-1.5cm, yshift=-.2cm] {$a_{=0}/{\small0}$} (q_2)
edge [below] node [xshift=.1cm,yshift=.4cm]{$a_{>0}/{\small+1}\ \ b_{=0}/{\small+}1$} (q_0);
\draw (5.4,15.9) node {$b_{=0}/{\small+1}$};
\draw (5.4,15.5) node {$b_{>0}/{\small+1}$};
\draw [gray!20] (-3,12)-- (-3,12) ;
\draw[<-] (2.33,16)-- (1.85,16);
\end{tikzpicture}}
\caption{\centering Learnt by "\minOCA".}
\label{anbnaPy}
\end{subfigure}
}
\hfill
{
   \begin{subfigure}[b]{0.35\textwidth}   
  \centering\scalebox{.5}{
\begin{tikzpicture}[shorten >=1pt,node distance=3cm,on grid,auto]
\tikzset{every path/.style={line width=.6mm}}\
\node[state] at (6,14) (L_0) {$L_0$};
\node[state] at (3,14) (L_6) {$L_6$};
\node[state] at (4,12)(L_7){$L_7$};
\node[state] at (5,10)(L_1){$L_1$};
\node[state] at (3,10)(L_3){$L_3$};
\node[state] at (3,6.5)(L_4){$L_4$};
\node[state] at (9,10)(L_5){$L_5$};
\node[state] at (7.5,8)(L_9){$L_9$};
\node[state] at (7,4)(L_8){$L_8$};
\node[state,accepting] at (9,6)(L_{10}){$L_{10}$};
\node[state] at (2,2)(L_2){$L_2$};
\path[->]
(L_0) edge [above] node {$a_{=0}/{\small0}$} (L_6)
(L_6) edge [right] node [xshift=-1cm, yshift=-.2cm]{$a_{=0}/{\small0}$} (L_7)
edge [bend left, right] node [xshift=.1cm]{$b_{=0}/{\small0}$} (L_5)
edge [bend left, right] node [xshift=0cm]{$b_{>0}/{\small-1}$} (L_9)
(L_7) edge [right] node{$b_{>0}/{\small0}$}(L_1)
edge [bend right,left] node[xshift=1.1cm]{$a_{=0}/{\small0}$}(L_3)
(L_1) edge [right] node[xshift=-.75cm, yshift=-.55cm]{$b_{>0}/{\small-1}$}(L_9)
edge [right] node[yshift=.2cm,xshift=.3cm]{$b_{=0}/{\small0}$}(L_5)
(L_5) edge [right] node {$a_{=0}/{\small0}$} (L_{10})
(L_3) edge [bend left,left] node [xshift=1.1cm]{$a_{=0}/{\small0}\ \ a_{>0}/{\small0}$} (L_4)
edge [bend right,left] node [xshift=1.3cm, yshift=-2cm]{$b_{=0}/{\small0}\ \ \ b_{>0}/{\small0}$} (L_2)
(L_2) edge [bend right,left] node [xshift=-.2cm, yshift=-2.2cm]{$b_{=0}/{\small0}\ \ \ b_{>0}/{\small0}$} (L_1)
(L_9) edge [left] node [xshift=1cm]{$b_{=0}/{\small0}\ \ b_{>0}/{\small0}$} (L_8)
(L_4) edge [left] node [xshift=1.2cm,yshift=-.9cm]{$b_{>0}/{\small0}$} (L_8)
edge [bend left] node {$a_{=0}/{\small+1}$} (L_6)
(L_8) edge [bend left] node [yshift=1.1cm, xshift=.2cm]{$b_{=0}/{\small0}\ \ \ b_{>0}/{\small0}$} (L_2);
\draw (1.1,9.9) node {$a_{>0}/{\small+1}$};
\draw (5.9,5.2) node {$b_{=0}/{\small0}$};
\draw (3.1,12.4) node {$a_{>0}/{\small0}$};
\draw (4.9,11.4) node {$b_{=0}/{\small0}$};
\draw (3.5,10.8) node {$a_{>0}/{\small0}$};
\draw[<-] (6.56,14)-- (7.15,14);
\end{tikzpicture}}
\caption{\centering Learnt by "\bps".}
\label{anbnaJav}
\end{subfigure}
}
\centering
\caption{The input "\droca" recognises the language $\{a^nb^na \mid n>0\}$. }
\label{ex1}
\end{figure}

\begin{figure}
  {
     \begin{subfigure}[b]{0.3\textwidth}   
  \centering\scalebox{.7}{
\begin{tikzpicture}[shorten >=1pt,node distance=3cm,on grid,auto]
\tikzset{every path/.style={line width=.4mm}}\
\node[state] at (0,6) (q_0) {$q_0$};
\node[state, accepting] at (-1.5,4.5) (q_1) {$q_1$};
\node[state] at (-3,3)(q_2){$q_2$};
\node[state] at (-3,1)(q_3){$q_3$};
\node[state] at (2 ,1)(q_4){$q_4$};
\path[->]
(q_0) edge [bend left, right] node[xshift=0cm,yshift=.3cm]  {$b_{=0}/{\small0}$} (q_4)
edge [above] node[xshift=-.4cm,yshift=-.15cm] {$a_{=0}/{\small0}$} (q_1)
(q_1) edge [loop right] node {$a_{=0}/{\small0}$} ()
edge [left] node[xshift=.1cm, yshift=.2cm] {$b_{=0}/{\small+1}$} (q_2)
(q_2) edge [loop right] node[xshift=0cm, yshift=.2cm]{$a_{>0}/{\small-1}$} ()
edge [below] node [yshift=.3cm, xshift=-.1cm]{$a_{=0}/{\small+1}\ \ b_{=0}/{\small0}$} (q_3)
(q_3) edge [below] node [yshift=.6cm]{$a_{=0}/{\small+1}\ \ a_{>0}/{\small-1}$} (q_4)
(q_4) edge [loop below] node[yshift=-.1cm]{$a_{>0}/{\small+1}$}();
\draw (-1.3,2.8) node {$b_{>0}/{\small-1}$};
\draw (2,-.7) node {$b_{=0}/{\small+1}$};
\draw (2,-1) node {$b_{>0}/{\small+1}$};
\draw (2,-.1) node {$a_{=0}/{\small+1}$};
\draw (2,-.1) node {$a_{=0}/{\small+1}$};
\draw (-.5,.6) node {$b_{=0}/{\small+1}\ \ b_{>0}/{\small-1}$};
\draw[<-] (-.53,6)-- (-1.05,6);
\end{tikzpicture}}
\caption{\centering Input "\droca".}
\label{inputdroca}
\end{subfigure}
}
\hfill
{
   \begin{subfigure}[b]{0.4\textwidth}   
  \centering\scalebox{.7}{
\begin{tikzpicture}[shorten >=1pt,node distance=3cm,on grid,auto]
\tikzset{every path/.style={line width=.4mm}}\
\node[state, accepting] at (0,9) (q_1) {$q_1$};
\node[state] at (-2,6) (q_3) {$q_3$};
\node[state] at (2.5,3)(q_2){$q_2$};
\node[state] at (3,10)(q_0){$q_0$};
\path[->]
(q_0) edge [loop right] node [yshift=.1cm]{$a_{>0}/{\small+1}$} ()
edge [bend right, below] node[xshift=-.6cm,yshift=.3cm]  {$b_{=0}/{\small0}$} (q_2)
edge [above] node[xshift=.1cm,yshift=.2cm] {$a_{=0}/{\small0}$} (q_1)
(q_1) edge [loop right,below] node {$a_{=0}/{\small0}$} ()
edge [above] node[xshift=-.5cm, yshift=0cm] {$b_{=0}/{\small+1}$} (q_3)
(q_2) edge [loop right] node[yshift=.1cm]{$a_{>0}/{\small-1}$} ()
edge [bend right, below] node[xshift=0cm,yshift=.3cm]  {$a_{=0}/{\small+1}\ \ b_{=0}/{\small+1}$} (q_0)
(q_3) edge [below] node [xshift=-.3cm, yshift=.4cm] {$a_{=0}/{\small+1}\ \ \  \ b_{=0}/{\small0}$} (q_2)
	edge [loop right] node [yshift=.1cm]{$a_{>0}/{\small-1}$} ();
\draw (-.4,5.8) node {$b_{>0}/{\small-1}$};
\draw (4.6,9.8) node {$b_{>0}/{\small+1}$};
\draw (4.1,2.8) node {$b_{>0}/{\small-1}$};
\draw[->] (3,10.9)-- (3,10.4);
\draw [gray!20] (-3,6)-- (-3,6) ;
\end{tikzpicture}}
\caption{\centering Learnt by "\minOCA".}
\label{newTestPython}
\end{subfigure}
}
\hfill
{
   \begin{subfigure}[b]{0.1\textwidth}   
  \centering\scalebox{.7}{
\begin{tikzpicture}[shorten >=1pt,node distance=3cm,on grid,auto]
\tikzset{every path/.style={line width=.4mm}}\
\node[state] at (0,5) (q_0) {$L_0$};
\node[state,accepting] at (0,0) (q_1) {$L_1$};
\path[->]
(q_0) edge [right] node {$a_{=0}/{\small0}$} (q_1)
(q_1) edge [loop right] node {$a_{=0}/{\small0}$} ();
\draw[->] (0,6)-- (0,5.5);
\end{tikzpicture}}
\caption{\centering Learnt by "\bps".}
\label{newTestJava}
\end{subfigure}
}
\caption{The input "\droca" recognises $\{w\in\{a,b\}^*\mid w \text{ does not contain a }b\}$.}
\label{ex2}
\end{figure}

\Cref{ex1,ex2}, give two examples of "\drocas" learnt by "\minOCA" and "\bps". In both the examples, "\minOCA" learns a minimal "counter-synchronised" "\droca" "equivalent" to the input. However, the "\droca" learnt by "\bps" is "equivalent" but is not "counter-synchronous" with respect to the input. It is also not complete. 

\clearpage
\section{Conclusion}
\label{sec:conclusion}
 In this paper, we presented a novel approach for active learning of "\drocas". We showed that a "\drocas" can be learnt using polynomially many queries with the help of a SAT solver. This is in contrast to the existing techniques that require exponentially many queries. Our algorithm learns a minimal "counter-synchronous" "\droca", which results in significantly smaller counter-examples on equivalence queries.  
  Additionally, the development of a specialised equivalence-checking algorithm for "counter-synchronous" "\drocas", with a time complexity of $\mathcal{O}(n^5)$, further contributes to the feasibility of our approach. For "\vocas", we optimised this equivalence-checking process to $\mathcal{O}(n^3)$, enabling efficient learning. 
We evaluated our algorithm against randomly generated "\drocas". 
The results indicate that our method significantly outperforms the existing method by Bruyère et al.~\cite{gaetan}. 

In future work, the proposed algorithm can be improved by finding better methods for identifying the minimal separating \dfa. The algorithm can also be applied to learn "\vocas", and the scalability of this approach warrants further investigation. Additionally, extending these ideas to learn more general models, such as visibly pushdown automata, represents a valuable direction for further research.
Another open problem is to determine whether active learning of "\drocas" can be done in polynomial time. This problem is open, even in the case of "\vocas". However, learning a minimal "\voca" cannot be done in polynomial time unless $\CF{P=NP}$. Exploring the possibility of finding a polynomial-time algorithm to obtain a polynomial approximation is worthy of further study.

\begin{credits}
\subsubsection{\ackname} {A.V. Sreejith would like to acknowledge the support by SERB for the project ``Probabilistic Pushdown Automata'' [MTR/2021/000788].}
\end{credits}

\addcontentsline{toc}{chapter}{Bibliography}
{\small
\bibliography{paperNew}
}

\end{document}